\pdfoutput=1
\documentclass{easychair}

\usepackage{times}
\usepackage{xspace}
\usepackage{amssymb, amsmath, amsthm}
\usepackage{enumerate}
\usepackage{paralist}
\usepackage{url}
\usepackage{lipsum}
\usepackage{siunitx}
\usepackage{xfrac}
\usepackage[disable]{todonotes}
\usepackage{multicol}

\usepackage[ruled]{algorithm}        % for the float (other options for the layout are "plain", "boxed")
\floatname{algorithm}{Algorithm}

\usepackage[noend]{algpseudocode}      % for the content [noend means no closing line for fors, ifs, whiles, etc. Indentation is enough.]

\algrenewcommand\algorithmicindent{1.0em}

\DeclareSymbolFont{matha}{OML}{txmi}{m}{it}% txfonts
\DeclareMathSymbol{\varv}{\mathord}{matha}{118}

\newif\iftautology
\tautologyfalse

\newcommand{\true}{1}
\newcommand{\false}{0}

\newcommand{\wrt}{w.r.t.\ }

\newcommand{\var}{\mathit{var}}

\newcommand{\ctrans}[1]{t}

\newcommand{\dom}{\mathit{dom}}
\newcommand{\range}{\mathit{ran}}

\newcommand{\mgu}{\mathit{mgu}\xspace}

\newcommand{\foeq}{\approx}
\newcommand{\foneq}{\not \foeq}

\newcommand{\eqaxioms}{\mathcal{E_{L}}}

\newcommand{\flatt}[1]{#1^{-}}
\newcommand{\lingeling}{\textsf{Lingeling}\xspace}
\renewcommand{\lingeling}{$\mathsf{Lingeling}\xspace$}
\newcommand{\minisat}{\textsf{MiniSat}\xspace}
\renewcommand{\minisat}{$\mathsf{MiniSat}\xspace$}
\newcommand{\glucose}{\textsf{Glucose}\xspace}
\renewcommand{\glucose}{$\mathsf{Glucose}\xspace$}

\newcommand{\abcdsat}{\textsf{abcdSAT}\xspace}
\renewcommand{\abcdsat}{$\mathsf{abcdSAT}\xspace$}

\newcommand{\fol}{FOL}
\newcommand{\foleq}{\fol$^{\foeq}$}
\newcommand{\cnf}{\mathit{cnf}}

\newcommand{\ppe}{PPE}
\newcommand{\ude}{UDE}

\newtheorem{definition}{Definition}
\newtheorem{theorem}{Theorem}
\newtheorem{lemma}[theorem]{Lemma}

\newtheorem{example}{Example}

\newcommand{\vampire}{{\sc Vampire}}

\title{Blocked Clauses in First-Order Logic\thanks{
This work has been supported by the Austrian Science Fund (FWF) under projects W1255-N23, S11408-N23, S11409-N23,
and the ERC Starting Grant 2014 SYMCAR 639270.}}
\author{Benjamin Kiesl\inst{1}, Martin Suda\inst{1}, Martina Seidl\inst{2}, Hans Tompits\inst{1}, Armin Biere\inst{2}}

\institute{Institute for Information Systems,
Vienna University of Technology, Austria \and
Institute for Formal Models and Verification, JKU Linz, Austria}

\authorrunning{Kiesl, Suda, Seidl, Tompits, and Biere}
\titlerunning{Blocked Clauses in First-Order Logic}

\begin{document}

\maketitle

\begin{abstract}
\noindent
%The abstract must not contain more than 200 words!
Blocked clauses provide the 
basis for powerful reasoning techniques
used in SAT, QBF, and DQBF solving.
Their definition, which relies on a simple syntactic
criterion, guarantees that they are both redundant and easy to find.
In this paper, we lift the notion of blocked clauses to first-order logic.
We introduce two types of blocked clauses, 
one for first-order logic with equality 
and the other for first-order logic without equality, 
and prove their redundancy.
In addition, we give a polynomial % time 
algorithm for checking whether a clause is blocked.
Based on our new notions of blocking,
we implemented a novel first-order preprocessing tool.
Our experiments showed that many
first-order problems in the TPTP library
contain a large number of blocked clauses.
Moreover, we observed that their elimination can improve
the performance of modern theorem provers, especially on satisfiable problem instances.
\end{abstract}

%-----------------------------------------------------------------------------------------------------------
\section{Introduction}
\label{sec:intro}
%-----------------------------------------------------------------------------------------------------------

Modern theorem provers often use dedicated
\emph{preprocessing methods} to speed up the proof search~\cite{DBLP:conf/fmcad/HoderKKV12,khasidashvili16_fo_predicate_elimination}. 
As most of these provers are based on proof systems that require 
formulas to be in conjunctive normal form (CNF), 
a wide range of established preprocessing methods performs simplifications on the CNF representation
of the input formula. 
Preprocessing on the CNF level is well explored for propositional 
logic~\cite{DBLP:conf/lpar/HeuleJB10} and has been successfully 
integrated into SAT solvers such 
as \minisat~\cite{een16_minisat}, 
\glucose~\cite{DBLP:conf/ijcai/AudemardLS13}, or 
\lingeling~\cite{Biere-SAT-Competition-2016-solvers}. 
But, although generalizations of several propositional preprocessing methods have 
been utilized by first-order theorem provers, one particularly successful concept
has, to the best of our knowledge, not yet found its way to first-order logic: 
the simple yet powerful concept of \emph{blocked clauses}~\cite{DBLP:journals/dam/Kullmann99}.  
In this paper, we address this issue and lift the notion of blocked clauses to first-order logic.

Informally,
a clause $C$ is \emph{blocked}
by one of its literals 
in a propositional CNF formula $F$ if all resolvents of $C$ upon this literal 
are tautologies~\cite{DBLP:journals/dam/Kullmann99}. 
A blocked clause is redundant in the sense that 
neither its deletion from nor its addition to $F$  affects the satisfiability or unsatisfiability 
of $F$. 
Blocked clauses provide the basis for the propositional preprocessing techniques of
\emph{blocked-clause elimination}~(BCE), \emph{blocked-clause addition} (BCA), and \emph{blocked-clause decomposition} (BCD).

Blocked-clause elimination considerably boosts solver performance by simulating 
several other, more complicated preprocessing techniques~\cite{DBLP:journals/jar/JarvisaloBH12}. 
But not only SAT solvers benefit from BCE;
even greater performance improvements are achieved when generalizations
of BCE are used for solving problems beyond  
the complexity class NP such as reasoning over quantified Boolean formulas~(QBF)~\cite{DBLP:journals/jair/HeuleJLSB15} 
or dependency quantified Boolean formulas~(DQBF)~\cite{DBLP:conf/sat/WimmerGNSB15}. 
When performed in a careful manner, however, also the \emph{addition} of certain small blocked clauses has shown to be useful~\cite{Inprocessing}.

Finally, blocked-clause decomposition~\cite{BlockedClauseDecomposition} is
a technique that splits a CNF formula into two parts that can in turn be solved
via blocked-clause elimination.
Applications of blocked-clause decomposition are, for instance,   
the identification of backbone variables, 
the detection of implied equivalences, and gate extraction. 
Moreover, the winner of the SAT-Race 2015 competition, \abcdsat~\cite{DBLP:journals/corr/Chen15o}, is\linebreak[4] based on blocked-clause decomposition.

The generalization of blocked clauses to the first-order case is not straightforward and poses several challenges; in particular,
the following two are crucial: 
First, the involvement of unification in first-order resolution brings some intricacies with it 
that are absent in propositional logic. 
A careful choice of resolvents is therefore essential
for ensuring the redundancy of blocked clauses.
Second, in the presence of equality, further problems
are caused by the fact
that Herbrand's Theorem has to be adapted 
in order to account for the peculiarities of equality.
Our approach successfully resolves these issues. 

The main contributions of this paper are the following: 
\begin{inparaenum}[(1)]
\item We present {blocked clauses} 
for first-order logic 
and prove their redundancy given that equality is not present.  
\item
 We introduce equality-blocked clauses, a refined notion of blocked clauses that guarantees redundancy even in the presence of equality.
\item We give a polynomial algorithm for deciding whether a 
clause is blocked.
\item To demonstrate one potential application of blocked 
   clauses, we implement a tool that performs
blocked-clause elimination and evaluate its impact on the 
performance of modern first-order theorem provers. 
\todo{Maybe two sentences about \cite{khasidashvili16_fo_predicate_elimination} as related work?: predicate elimination
is first order generalisation of the variable and clause elimination technique \cite{DBLP:conf/sat/EenB05}. Inspiration, bla bla. 
See also the note before Lemma \ref{thm:f_assignments_can_be_extended}.}
\end{inparaenum}

This paper is structured as follows. After introducing the necessary preliminaries in Section~\ref{sec:prelim}, we shortly recapitulate the propositional 
notion of blocked clauses and lift it to first-order logic in 
Section~\ref{sec:bc_fol_without_equality}. In Section~\ref{sec:bc_fol_with_equality}, we introduce equality-blocked
clauses and prove that they are redundant even when the equality predicate  
is present. We discuss the complexity of deciding the blockedness of a 
clause in Section~\ref{sec:complexity}. Finally, in Section~\ref{sec:bce},
we present our implementation of 
blocked-clause elimination,
relate it to other first-order preprocessing techniques, 
and evaluate its impact on first-order theorem provers.  

%-----------------------------------------------------------------------------------------------------------
\section{Preliminaries}
\label{sec:prelim}
%-----------------------------------------------------------------------------------------------------------

We assume the reader to be familiar with the basics of %propositional and 
first-order logic.
% (see, e.g., \cite{fitting96_fo_logic}).
As usual, formulas 
of a first-order language $\cal{L}$
are built using predicate symbols, function symbols, and constants from some given 
alphabet 
 together with logical connectives, quantifiers, and variables.
We use the letters $a,b,c,\ldots$ for constants and $x,y,z,u,\varv,\ldots$ for variables (possibly with 
subscripts). 
%The use of predicate and function symbols should be clear from the context.
%A \emph{term} is a variable, constant, or an expression of the form $f(t_1, \dots, t_n)$ with $f$ being a function symbol and $t_1, \dots, t_n$ being terms.
The equality predicate symbol $\foeq$ is used in infix notation and we write $x \foneq y$ for $\neg (x \foeq y)$.
An expression (i.e., a term, literal, formula, etc.) is \emph{ground} if it contains no variables.

%If not stated otherwise, formulas are assumed to be in conjunctive normal form (CNF), i.e., 
%a conjunction of clauses, each of which is a universally quantified disjunction of literals.
%A formula is in \emph{conjunctive normal form} (\emph{CNF}) if it is a universally quantified conjunction of clauses, 
%each of which is a disjunction of literals.

A \emph{literal} is an atom or the negation of an atom 
and a disjunction of literals is a \emph{clause}. 
For a literal $L$ and an atom $P$, we define
$\bar L = \neg P$ if $L = P$ and $\bar L = P$ if $L = \neg P$. 
In the former case, $L$ 
is of \emph{positive polarity}; in the latter case, it is of \emph{negative polarity}.
A formula is in \emph{conjunctive normal form} (\emph{CNF}) if it is a conjunction of clauses. 
% A formula in \emph{conjunctive normal form} (\emph{CNF}) is a conjunction of clauses. 
% A clause is a disjunction of literals and a literal is a (possibly negated) atom.
W.l.o.g., clauses are assumed to be variable disjoint.
Variables of a CNF formula are implicitly universally 
quantified.
We treat CNF formulas as sets of clauses and
clauses as multisets of literals.
If not stated otherwise, we assume formulas to be in CNF.
%When writing a clause, we usually omit quantifiers.
\iftautology\else
A clause is a \emph{tautology} if it contains both $L$ and $\bar L$ for some literal $L$.
\fi
%For a literal $L(\dots)$ with predicate symbol $P$, 
%$\bar L(\dots) = \neg P(\dots)$ if $L(\dots) = P(\dots)$ and $\bar L(\dots) = P(\dots)$ if $L(\dots) = \neg P(\dots)$. 
%In the former case, $L(\dots)$ 
% 
% For a literal $L$, 
% $\bar L = \neg P$ if $L = P$ and $\bar L = P$ if $L = \neg P$. 
% In the former case, $L$ 
% is of \emph{positive polarity}; in the latter case it is of \emph{negative polarity}.

%Regarding the semantics, %assume familiarity with the 
We use the standard notions of \emph{interpretation}, \emph{model}, \emph{validity}, \emph{satisfiability}, \emph{logical equivalence}, and \emph{satisfiability equivalence}.
The predicate symbol $\foeq$ is special as it must be \textbf{interpreted} as the identity relation over the domain under consideration.
\iftautology
By a \emph{tautology}, we mean a formula that is satisfied by every interpretation.
\fi
A \emph{propositional assignment} is a mapping from ground atoms to the truth values~$\true$~(\emph{true}) and $\false$ (\emph{false}).
Accordingly, a set of ground clauses is \emph{propositionally satisfiable} if there exists a propositional assignment that satisfies $F$ under the usual semantics for the logical connectives.
An assignment $\alpha'$ is obtained of an assignment $\alpha$ by \emph{flipping} the truth value of a literal $L$ if $\alpha'$ agrees with $\alpha$ on all atoms except for that of $L$ to which it assigns the opposite truth value.
We sometimes write 
propositional assignments as sequences of literals
where a positive (negative) polarity of a literal indicates that
its corresponding atom is assigned to true (false,~respectively).

\iffalse
We distinguish between first-order logic without equality~(\fol) and with equality (\foleq).
In the latter, 
every interpretation is required to interpret the predicate symbol $\foeq$ as the identity relation over its domain.
As usual, we write $x \foneq y$ for $\neg (x \foeq y)$.
\fi

A \emph{substitution} is a mapping from variables to terms that agrees with the identity function on all but finitely many variables.
Let $\sigma$ be a substitution.
The domain, $\dom(\sigma)$, of $\sigma$ is the set of variables for which $\sigma(x) \neq x$.
The range, $\range(\sigma)$, of $\sigma$ is the set $\{\sigma(x) \mid x \in \dom(\sigma)\}$. 
A substitution is \emph{ground} if its range consists only of ground terms.
As common, $E\sigma$ denotes the result of applying $\sigma$ to the expression~$E$.
If $E\sigma$ is ground, it is a \emph{ground instance} of $E$.
Juxtaposition of substitutions denotes their composition, i.e., $\sigma\tau$ stands for $\tau \circ \sigma$.
The substitution $\sigma$ is a \emph{unifier} of the expressions $E_1, \dots, E_n$ if $E_1\sigma = \dots =  E_n\sigma$.
For substitutions $\sigma$ and $\tau$, we say that $\sigma$ is \emph{more general} than $\tau$ if there exists
a substitution $\lambda$ such that $\sigma\lambda = \tau$.
Furthermore, $\sigma$ is a \emph{most general unifier} ($\mgu$) of $E_1, \dots, E_n$ if, for every
unifier $\tau$ of $E_1, \dots, E_n$, $\sigma$ is more general than $\tau$.
In the rest of the paper, 
we make use of two popular variants of Herbrand's Theorem
% ~\cite{herbrand30_recherches}:
(cf.~\cite{fitting96_fo_logic}):

%\begin{theorem}\label{thm:herbrands_theorem}
%A CNF formula $F$ without equality is satisfiable
%iff 
%every finite set of ground instances of clauses in $F$ is propositionally satisfiable.
%\end{theorem}
\begin{theorem}\label{thm:herbrands_theorem}
A %CNF 
formula $F$ that does not contain the equality predicate is satisfiable
iff 
every finite set of ground instances of clauses in $F$ is propositionally satisfiable.
\end{theorem}

\noindent
%In the presence of equality, a formula $F$ is satisfiable iff $F \cup \eqaxioms$ is satisfiable without the restriction that $\foeq$ is interpreted as the identity relation over the given domain, 
%where $\eqaxioms$ denotes the following set of equality axioms for the language $\cal{L}$ under consideration
%(cf.~\cite{fitting96_fo_logic}):
%\begin{enumerate}[(E1)]
%	\item $x \foeq x$;
%	\item for each $n$-ary function symbol $f$ in $\cal{L}$, 
%	$x_1 \foneq y_1 \lor \dots \lor x_n \foneq y_n \lor f(x_1,\dots,x_n) \foeq f(y_1,\dots,y_n)$; and
%	\item for each $n$-ary predicate symbol $P$ in $\cal{L}$, 
%	$x_1 \foneq y_1$ $\lor \dots \lor x_n \foneq y_n \lor \neg P(x_1,\dots,x_n) \lor P(y_1,\dots,y_n)$.
%\end{enumerate}
%
%\noindent
%Hence,  the following variant of Herbrand's Theorem follows in the presence of equality:
%
%\begin{theorem}\label{thm:herbrand_with_equality}
%A CNF formula $F$ with equality is satisfiable
%iff 
%every finite set of ground instances of clauses in \mbox{$F \cup \eqaxioms$} is propositionally satisfiable.
%\end{theorem}
Furthermore, a formula $F$ that contains the equality predicate is satisfiable iff $F \cup \eqaxioms$ is satisfiable without the restriction that $\foeq$ must be interpreted as the identity relation,
where $\eqaxioms$ denotes the following set of \emph{equality axioms} for the language $\cal{L}$ under consideration%
~(cf.~\cite{fitting96_fo_logic}):
\begin{enumerate}[(E1)]
	\item $x \foeq x$;
	\item for each $n$-ary function symbol $f$ in $\cal{L}$, 
	$x_1 \foneq y_1 \lor \dots \lor x_n \foneq y_n \lor f(x_1,\dots,x_n) \foeq f(y_1,\dots,y_n)$; 
	\item for each $n$-ary predicate symbol $P$ in $\cal{L}$, 
	$x_1 \foneq y_1$ $\lor \dots \lor x_n \foneq y_n \lor \neg P(x_1,\dots,x_n) \lor P(y_1,\dots,y_n)$.
\end{enumerate}

\noindent
Hence,  the following variant of Herbrand's Theorem for formulas with equality follows:

\begin{theorem}\label{thm:herbrand_with_equality}
A %CNF 
formula $F$ that contains the equality predicate is satisfiable
iff 
every finite set of ground instances of clauses in \mbox{$F \cup \eqaxioms$} is propositionally satisfiable.
\end{theorem}

\noindent
Next, we formally introduce the redundancy of clauses.
Intuitively, a clause $C$ is redundant \wrt a formula $F$ if neither its addition to $F$ 
nor its removal from $F$ changes the satisfiability or unsatisfiability of $F$~\cite{heule16_solutionvalidationqbf}:%
%\footnote{
%We note that we are using a different notion of redundancy 
%than the one usually employed when formalizing the theorem proving process
%within the context of ordered resolution \cite{DBLP:books/el/RV01/BachmairG01}.}

\begin{definition}\label{def:clause_redundancy}
A clause $C$ is \emph{redundant} \wrt a %CNF 
formula $F$ if $F \setminus \{C\}$ and $F \cup \{C\}$ are satisfiability equivalent. 
\end{definition}

\noindent
Note that this notion of redundancy does not require \emph{logical} equivalence of 
$F \setminus \{C\}$ and $F \cup \{C\}$ and that it 
is different from the Bachmair-Ganzinger notion of redundancy that 
is usually employed within the context of ordered resolution~\cite{DBLP:books/el/RV01/BachmairG01}.
It provides the basis for both clause elimination and clause addition procedures.
Note also that the redundancy of a clause $C$ \wrt a formula $F$ can be shown by proving that the
satisfiability of $F \setminus \{C\}$ implies the satisfiability of $F \cup \{C\}$.

Finally, 
% we introduce the resolution rule. 
given two clauses 
$C = L_1 \lor \dots \lor L_k \lor C'$ and $D = N_1 \lor \dots \lor N_l \lor D'$ such that the
literals $L_1, \dots, L_k, \bar N_1, \dots, \bar N_l$ are unifiable by an $\mgu$~$\sigma$, 
the clause $C'\sigma \lor D'\sigma$ is said to be a \emph{resolvent}
of $C$ and $D$.
If $k = l = 1$, it is a \emph{binary resolvent} of $C$ and $D$ upon $L_1$.

\section{Blocked Clauses 
%in FOL without Equality
% in \fol
}
\label{sec:bc_fol_without_equality}
%-----------------------------------------------------------------------------------------------------------

In this section, we first recapitulate the notion of blocked clauses used 
in propositional logic. 
We then illustrate complications that arise when
lifting blocked clauses to first-order logic.
As main result of the section, we introduce blocked clauses for first-order logic %without equality
and prove that they are redundant if the %formula under consideration does not contain the equality predicate.
equality predicate is not present.
Throughout this section, we therefore consider only clauses and formulas without the equality predicate.

In propositional logic, a clause $C$ is %said to be 
\emph{blocked} by a literal $L \in C$ in a CNF formula $F$ if 
all binary resolvents of $C$ upon $L$ with clauses 
from $F \setminus \{C\}$ are tautologies.
%\todo{Benjamin: I'm still not happy with the next sentence, I don't like the ``Where by''.} 
%By \emph{all binary resolvents} we mean all binary resolvents with clauses in $F \setminus \{C\}$.
%We say that 
A clause $C$ is blocked in a formula~$F$ if $C$ is blocked in~$F$ by one (or more) of its
literals.
%Consider, for instance, the following example:

%\begin{example}\label{ex:prop_blocked_clause}
%Given formula
%$F = \{$\mbox{$P \lor \neg Q$}, \mbox{$\neg Q \lor R$}$\}$, the clause
%$C = \neg P \lor Q$ is blocked by $\neg P$ in $F$: 
%There is only one resolvent of $C$ upon $\neg P$, namely the tautology $Q \lor \neg Q$, 
%obtained by resolving with $P \lor \neg Q$.
%\end{example}
\begin{example}\label{ex:prop_blocked_clause}
The clause $C = \neg P \lor Q$ is blocked by $\neg P$ in
$F = \{$\mbox{$P \lor \neg Q$}, \mbox{$\neg Q \lor R$}$\}$: 
The only resolvent of $C$ upon $\neg P$ is the tautology $Q \lor \neg Q$, 
obtained by resolving with $P \lor \neg Q$.
\end{example}

\noindent
%We say that a clause $C$ is blocked in a formula~$F$ if one or more of its literals block it in~$F$.
Under the restriction---common in propositional logic---that clauses must not contain multiple occurrences 
of the same literal, 
it can be shown that 
blocked clauses are redundant:
Let $C$ be blocked by $L \in C$ in a formula $F$. 
Then, every assignment that satisfies $F \setminus \{C\}$ but falsifies~$C$ can be turned into 
a satisfying assignment of $C$ by simply flipping the truth value of~$L$, i.e., by inverting the truth value of its 
atom\iffalse~(see \cite{DBLP:conf/sat/JarvisaloB10} for details)\fi. 
This flipping does not  falsify any of the clauses in $F \setminus \{C\}$ 
that contains~$\bar L$, because of 
the fact that every binary resolvent of $C$ upon $L$ is a tautology:
% guarantees that these clauses stay satisfied: 
A clause that contains $\bar L$ either is itself a tautology or 
it contains a literal $R \neq \bar L$ such that $\bar R \in C$.
In the latter case, since $C$ and thus $\bar R$ was assumed to be false before the flipping of the truth value of $L$, 
$R$ also stays true afterwards.

\begin{example}
Consider again $C$ and $F$ from Example \ref{ex:prop_blocked_clause}. 
The assignment $P\neg QR$ satisfies $F \setminus \{C\}$ but falsifies $C$. 
By flipping the truth value of $\neg P$, we obtain the assignment $\neg P \neg Q R$ 
that satisfies $F \cup \{C\}$. 
The only clause that could have possibly been falsified, namely $P \lor \neg Q$, stays true since it contains $\neg Q$ 
which was true before the flipping. 
\iffalse
The occurrence of $\neg Q$ in $P \lor \neg Q$ is also the reason why the resolvent of $C$ and $P \lor \neg Q$ is a tautology.
\fi
\end{example}

\noindent
As can be seen in the next example, 
%the above argument 
%about the redundancy of blocked clauses
%breaks down 
redundancy is not guaranteed
when clauses are allowed to contain multiple occurrences of the same literal.
Although the example might seem pathological at first, it will help to illustrate an inherent complication
arising in first-order logic:

\begin{example}\label{ex:multiple_occurrences}
Let $C = P \lor P$ and $F = \{\neg P \lor \neg P\}$. 
Clearly, $F \setminus \{C\} = F$ is satisfiable whereas $F \cup \{C\}$ is not. 
There is one binary resolvent of $C$ upon $P$, 
namely the tautology $P \lor \neg P$, hence $C$ is blocked by $P$ in $F$.
However, turning a satisfying assignment of $F$ \textup{(}i.e., one that falsifies $P$\textup{)} into one
of $C$ by flipping the truth value of $P$ falsifies $\neg P \lor \neg P$.
\end{example}

\noindent
In first-order logic, the requirement that all binary resolvents of $C$ upon $L$ are
\iftautology
tautologies
\else
valid%
\footnote{As common in first-order logic, we use the notion of \emph{validity} instead of \emph{tautologyhood}. 
In the absence of equality, a clause is valid if and only if it contains two complementary literals~$L, \bar L$.}
\fi%
fails to guarantee redundancy, 
even when clauses are not allowed to contain multiple occurrences of the same literal.
The reason is that similar
issues as in Example~\ref{ex:multiple_occurrences}
might occur on the ground level after certain literals are instantiated through unification:

%\begin{example}\label{ex:bc_fo_unsound}
%Let $C = P(x,x) \lor P(c,x)$ and $F = \{$\mbox{$\neg P(y,y) \lor \neg P(c,y)$}$\}$. 
%The literal $P(x,x)$ blocks $C$ in~$F$ since there are \textup{(}modulo renaming\textup{)} 
%two possible resolvents of $C$ upon $P(x,x)$ that are both tautologies: 
%The resolvent \mbox{$P(c,y) \lor \neg P(c,y)$}, obtained by using the $\mgu$ $\{x \mapsto y\}$ of $P(x,x)$ and $P(y,y)$, 
%and the resolvent \mbox{$P(c,c) \lor \neg P(c,c)$}, obtained by using the $\mgu$ $\{x \mapsto c, y \mapsto c\}$ of $P(x,x)$ and $P(c,y)$.
%However, the formula \mbox{$F \setminus \{C\}$} $ = F$ is clearly satisfiable 
%whereas \mbox{$F \cup \{C\}$} is not.
%To see this, observe that the ground instances $P(c,c) \lor P(c,c)$ and $\neg P(c,c) \lor \neg P(c,c)$ of $C$ 
%and $\neg P(y,y) \lor \neg P(c,y)$, respectively, are not simultaneously satisfiable.
%Since $P(y,y)$ and $P(c,y)$ both unify with $P(x,x)$ on the ground level, 
%we basically face the same problem as in Example \ref{ex:multiple_occurrences}.
%\end{example}

\newcommand{\xone}{x}
\newcommand{\xtwo}{y}
\newcommand{\yone}{u}
\newcommand{\ytwo}{\varv}

\begin{example}\label{ex:bc_fo_unsound}
Consider $C = P(\xone,\xtwo) \lor P(\xtwo,\xone)$ and $F = \{$\mbox{$\neg P(\yone,\ytwo) \lor \neg P(\ytwo,\yone)$}$\}$. 
Two binary resolvents can be derived from $C$ upon $P(\xone,\xtwo)$ and both are 
\iftautology
tautologies:
\else
valid:
\fi
The resolvent \mbox{$P(\ytwo,\yone) \lor \neg P(\ytwo,\yone)$}, obtained by using the $\mgu$ $\{\xone \mapsto \yone, \xtwo \mapsto \ytwo\}$ of $P(\xone,\xtwo)$ and $P(\yone,\ytwo)$, 
and the resolvent \mbox{$P(\yone,\ytwo) \lor \neg P(\yone,\ytwo)$}, 
obtained by using the $\mgu$ $\{\xone \mapsto \ytwo, \xtwo \mapsto \yone\}$ of $P(\xone,\xtwo)$ and $P(\ytwo,\yone)$.
However, the formula \mbox{$F \setminus \{C\}$} $ = F$ is clearly satisfiable 
whereas \mbox{$F \cup \{C\}$} is not.
To see this, observe that there exists no satisfying assignment for the two
ground instances $P(c,c) \lor P(c,c)$ and $\neg P(c,c) \lor \neg P(c,c)$ of $C$ 
and $\neg P(\yone,\ytwo) \lor \neg P(\ytwo,\yone)$, respectively.
Since $P(\yone,\ytwo)$ and $P(\ytwo,\yone)$ both unify with $P(\xone,\xtwo)$ on the ground level, 
we face the same problem as in Example \ref{ex:multiple_occurrences}.
\end{example}

%------------ BEGIN:  DON'T REMOVE -------------------------
\iffalse
\begin{example}\label{ex:bc_fo_unsound}
Let $C = P(x,x) \lor P(c,x)$ and $F = \{$\mbox{$\neg P(y,y) \lor \neg P(c,y)$}$\}$. 
There are \textup{(}modulo renaming\textup{)} 
two binary resolvents of $C$ upon $P(x,x)$ both of which are 
\iftautology
tautologies:
\else
valid:
\fi
The resolvent \mbox{$P(c,y) \lor \neg P(c,y)$}, obtained by using the $\mgu$ $\{x \mapsto y\}$ of $P(x,x)$ and $P(y,y)$, 
and the resolvent \mbox{$P(c,c) \lor \neg P(c,c)$}, obtained by using the $\mgu$ $\{x \mapsto c, y \mapsto c\}$ of $P(x,x)$ and $P(c,y)$.
However, the formula \mbox{$F \setminus \{C\}$} $ = F$ is clearly satisfiable 
whereas \mbox{$F \cup \{C\}$} is not.
To see this, observe that there exists no satisfying assignment for the two
ground instances $P(c,c) \lor P(c,c)$ and $\neg P(c,c) \lor \neg P(c,c)$ of $C$ 
and $\neg P(y,y) \lor \neg P(c,y)$, respectively.
Since $P(y,y)$ and $P(c,y)$ both unify with $P(x,x)$ on the ground level, 
we face the same problem as in Example \ref{ex:multiple_occurrences}.
\end{example}
\fi
%------------- END:  DON'T REMOVE --------------------------

\noindent
Such examples are often used for illustrating that binary resolution
alone does not guarantee completeness of the resolution principle 
(see, e.g., \cite{fitting96_fo_logic} or 
\cite{leitsch97_resolution_calculus}).
Analogously, we have to test the
\iftautology
tautologyhood
\else
validity
\fi
of more than just its binary resolvents 
in order to guarantee the redundancy of a clause.
But there is no need to test all possible resolvents.
As we will see, it is enough to consider the following ones: 

\begin{definition}
Let $C = L \lor C'$ and $D = N_1 \lor \dots \lor N_l \lor D'$ with $l>0$ be clauses such that the 
literals $L, \bar N_1, \dots, \bar N_l$ are unifiable by an $\mgu$~$\sigma$.
Then, $C'\sigma \lor D'\sigma$ is called \emph{$L$-resolvent} of $C$ and $D$.
\end{definition}

\begin{definition}\label{def:fo_bc_noneq}
A clause $C$ is \emph{blocked} by a literal $L \in C$ in a formula $F$ if 
all $L$-resolvents of $C$ % \textup{(}
with clauses in $F \setminus \{C\}$ %\textup{)} 
are 
\iftautology
tautologies.
\else
valid.
\fi
\end{definition}

\noindent
For instance, in Example~\ref{ex:bc_fo_unsound}, the clause $C$ is \emph{not} blocked by $L = P(\xone,\xtwo)$ in~$F$.
In addition to the two 
\iftautology
tautological 
\else
valid
\fi
binary resolvents---which are both $L$-resolvents---already considered in the example,
there is another 
\iftautology
(non-tautological) 
\fi
$L$-resolvent of $C$ and $\neg P(\yone,\ytwo) \lor \neg P(\ytwo,\yone)$,
namely $P(\xone,\xone)$%
\iftautology
\else
\ (which is not valid)%
\fi
, obtained by unifying $P(\xone,\xtwo)$, $P(\yone,\ytwo)$, and $P(\ytwo,\yone)$ via the $\mgu$ $\{\xtwo \mapsto \xone, \yone \mapsto \xone, \ytwo \mapsto \xone\}$.
Example \ref{ex:fo_bc_noneq} shows a clause that is blocked according to Definition~\ref{def:fo_bc_noneq}:

\begin{example}\label{ex:fo_bc_noneq}
Let $C = P(x,y) \lor \neg P(y,x) \lor Q(b)$ and $F = \{\neg P(a,b) \lor P(b,a), \neg Q(b)\}$. 
Then, $L = P(x,y)$ blocks $C$ in $F$ since there is only a single $L$-resolvent of $C$ upon $P(x,y)$, 
namely $\neg P(b,a) \lor Q(b) \lor P(b,a)$, 
obtained by using the $\mgu{}$ $\{x \mapsto a, y \mapsto b\}$ of the literals $P(x,y)$ and $P(a,b)$,
and this resolvent is
\iftautology
tautological.
\else
valid.
\fi
\end{example}

%------------ BEGIN:  DON'T REMOVE -------------------------
\iffalse
\noindent
For instance, in Example~\ref{ex:bc_fo_unsound}, the clause $C$ is \emph{not} blocked by $L = P(x,x)$ in~$F$.
In addition to the two 
\iftautology
tautological 
\else
valid
\fi
binary resolvents---which are both $L$-resolvents---already considered in the example,
there is another 
\iftautology
(non-tautological) 
\fi
$L$-resolvent of $C$ and $\neg P(y,y) \lor \neg P(c,y)$,
namely $P(c,c)$%
\iftautology
\else
\ (which is not valid)%
\fi
, obtained by unifying $P(x,x)$, $P(y,y)$, and $P(c,y)$ via the $\mgu$ $\{x \mapsto c, y \mapsto c\}$.
Example \ref{ex:fo_bc_noneq} shows a clause that is blocked according to Definition~\ref{def:fo_bc_noneq}:

\begin{example}\label{ex:fo_bc_noneq}
Let $C = P(x,y) \lor \neg P(y,x) \lor Q(b)$ and $F = \{\neg P(a,b) \lor P(b,a), \neg Q(b)\}$. 
Then, $L = P(x,y)$ blocks $C$ in $F$ since there is only a single $L$-resolvent of $C$ upon $P(x,y)$, 
namely $P(a,b) \lor \neg P(b,a) \lor Q(b) \lor P(b,a)$, 
obtained by using the $\mgu$ $\{x \mapsto a, y \mapsto b\}$ of the literals $P(x,y)$ and $P(a,b)$,
and this resolvent is
\iftautology
tautological.
\else
valid.
\fi
\end{example}
\fi
%------------- END:  DON'T REMOVE --------------------------

\noindent
%In propositional logic, given a satisfying assignment of \mbox{$F \setminus \{C\}$} that falsifies $C$ (with $C$ being blocked in $F$),
%we can satisfy $C$ without falsifying any clauses in $F \setminus \{C\}$ by flipping the truth value of the blocking literal.
%Something similar can be done with ground instances of blocked clauses in first-order logic.
Similar to the propositional case, where a satisfying assignment of $F \setminus \{C\}$ (with $C$ being blocked in $F$)
can be turned into one of $F \cup \{C\}$ by flipping
the truth value of the blocking literal, we can satisfy ground
instances of blocked clauses in first-order logic.
For instance, in Example~\ref{ex:fo_bc_noneq}, 
the assignment $\alpha = \neg P(a,b) P(b,a) \neg Q(b)$ satisfies \mbox{$F \setminus \{C\}$}
(which is already ground) but falsifies the ground instance $P(a,b) \lor \neg P(b,a) \lor Q(b)$ of~$C$.
By flipping the truth value of $P(a,b)$ we obtain 
$\alpha' = P(a,b) P(b,a) \neg Q(b)$---a satisfying assignment of this ground instance that still satisfies~$F$. 

%\begin{lemma}\label{thm:fo_flipping_works}
%Let $C$ be blocked by $L$ in $F$, and $\alpha$ a propositional assignment that falsifies a ground instance $C\lambda$ of~$C$. 
%Then, the assignment $\alpha'$, obtained from $\alpha$ by flipping the truth value of $L\lambda$, 
%still satisfies all the ground instances of clauses in $F \setminus \{C\}$ that have been satisfied by $\alpha$.
%\end{lemma}
\begin{lemma}\label{thm:fo_flipping_works}
Let $C$ be blocked by $L$ in $F$, and $\alpha$ a propositional assignment that falsifies a ground instance $C\lambda$ of~$C$. 
Then, the assignment $\alpha'$, obtained from $\alpha$ by flipping the truth value of $L\lambda$, 
satisfies all the ground instances of clauses in $F \setminus \{C\}$ that are satisfied by $\alpha$.
\end{lemma}

\begin{proof}
Let $D\tau$ be a ground instance of a clause $D \in F \setminus \{C\}$ and suppose $\alpha$ satisfies $D\tau$. 
If $D\tau$ does not contain $\bar L\lambda$ it is trivially satisfied by $\alpha'$. 
Assume therefore that $\bar L\lambda \in D\tau$ and 
let $N_1, \dots, N_l$ be all the literals in $D$ such that $N_i\tau = \bar L\lambda$ for $1 \leq i \leq l$.
Then, the substitution $\lambda\tau = \lambda \cup \tau$ 
(note that $C$ and $D$ are variable disjoint by assumption)
is a unifier of $L, \bar N_1, \dots, \bar N_l$.
Since $C$ is blocked by $L$ in $F$, the $L$-resolvent $(C \setminus \{L\})\sigma \lor (D \setminus \{N_1, \dots, N_l\})\sigma$,
with $\sigma$ being an $\mgu$ of $L, \bar N_1, \dots, \bar N_l$, is
\iftautology
a tautology.
\else
valid.
\fi
As $\sigma$ is most general, it follows that $\sigma\gamma = \lambda\tau$ for some substitution~$\gamma$. 
Hence,
\[
\begin{array}{llll}
 	&(C \setminus \{L\})\sigma\gamma &\lor\ (D \setminus \{N_1, \dots, N_l\})\sigma\gamma\\
=	&(C \setminus \{L\})\lambda\tau 	&\lor\ (D \setminus \{N_1, \dots, N_l\})\lambda\tau\\
=	&(C \setminus \{L\})\lambda 		&\lor\ (D \setminus \{N_1, \dots, N_l\})\tau
\end{array}
\]
% 	$(C \setminus \{L\})\sigma\gamma \lor (D \setminus \{N_1, \dots, N_l\})\sigma\gamma
%=	(C \setminus \{L\})\lambda\tau \lor (D \setminus \{N_1, \dots, N_l\})\lambda\tau
%=	(C \setminus \{L\})\lambda 		\lor (D \setminus \{N_1, \dots, N_l\})\tau$
is 
\iftautology
a tautology.
\else
valid.
\fi
Thus, since $\alpha$ falsifies $C\lambda$, it must satisfy a literal $L'\tau \in (D \setminus \{N_1, \dots, N_l\})\tau$. 
But, as all the literals in $(D \setminus \{N_1, \dots, N_l\})\tau$ are different from $\bar L\lambda$,
flipping the truth value of $L\lambda$ does not affect the truth value of $L'\tau$.
It follows that $\alpha'$ satisfies $L'\tau$ and thus it satisfies~$D\tau$.
\end{proof}

\noindent
A falsified ground instance $C\lambda$ of $C$ can therefore be satisfied 
without falsifying any ground instances of clauses in \mbox{$F\setminus\{C\}$} by simply flipping the truth value of $L\lambda$.
Still, it could happen that this flipping falsifies other ground instances of~$C$ itself,
namely those in which the only satisfied literals are complements of $L\lambda$. 
%Still, the flipping could possibly falsifiy other ground instances of~$C$ if they contain complements of $L\lambda$.
As it turns out, this is not a serious problem. 
Consider the following example:

\begin{example}
Given $C$ and $F$ from Example \ref{ex:fo_bc_noneq}, let $P(a,b) \lor \neg P(b,a) \lor Q(b)$ and $P(b,a) \lor \neg P(a,b) \lor Q(b)$ be 
the two 
\iftautology
non-tautological 
\fi
ground instances\footnote{With respect to the (here) finite Herbrand universe $\{a,b\}$.} 
of~$C$ \iftautology\else that are not valid\fi.
As shown above, the satisfying assignment $\alpha = \neg P(a,b)P(b,a) \neg Q(b)$ of $F$
can be turned into the satisfying assignment $\alpha' = P(a,b)P(b,a) \neg Q(b)$ of $P(a,b) \lor \neg P(b,a) \lor Q(b)$ by flipping the truth value of $P(a,b)$.
Now, $\alpha'$ falsifies the other ground instance $P(b,a) \lor \neg P(a,b) \lor Q(b)$ of $C$.
% \textup{(}since $\neg P(a,b)$ was its only literal satisfied by $\alpha$\textup{)}.

But, by flipping the truth value of yet another instance of the blocking literal---this time that of $P(b,a)$---we can also satisfy $P(b,a) \lor \neg P(a,b) \lor Q(b)$.
We don't need to worry that this flipping falsifies $P(a,b) \lor \neg P(b,a) \lor Q(b)$ again---the
instance $P(a,b)$ of the blocking literal cannot be falsified by making a literal of the form $P(\dots)$ true.
The resulting assignment $\alpha'' = P(a,b) P(b,a) \neg Q(b)$
is then a satisfying assignment of all ground instances of clauses in $F \cup \{C\}$.
\end{example} 

\noindent
The proof of the following lemma is based on the idea of repeatedly making instances of the blocking literal true.
We remark that---thanks to this lemma---the definition of a blocked clause % (Definition~\ref{def:fo_bc_noneq}) 
can safely ignore resolvents of the clause $C$ with itself. 
It is not a priori obvious that these resolvents can be ignored when lifting the propositional notion
since ``on the ground level'' two different instances of $C$ may become premises of a resolution step.
(For this exact reason, Khasidashvili and Korovin \cite{khasidashvili16_fo_predicate_elimination}
restrict their attention to non-self-referential predicates with their predicate elimination technique, 
a lifting of variable elimination \cite{DBLP:conf/sat/EenB05}.)
%However, in the case of blocked clause, the self-resolution check can actually be skipped.

\begin{lemma}\label{thm:f_assignments_can_be_extended}
Let $C$ be blocked in $F$ 
and let $F'$ and $F_C$ be finite sets of ground instances of clauses in $F \setminus \{C\}$ and~$\{C\}$, respectively. 
Then, every assignment that propositionally satisfies $F'$ can be turned into one that satisfies $F' \cup F_C$.
\end{lemma}

\begin{proof}
\newcommand{\blitterms}{t_1, \dots, t_n}
\newcommand{\blit}{L}
\newcommand{\nblit}{\bar L}
Let $C$ be blocked by $\blit$ in $F$ and let $\alpha$ be a satisfying assignment of $F'$.
Assume furthermore that $\alpha$ does not satisfy $F_C$, i.e.,
there exist ground instances of $C$ that are falsified by $\alpha$.
By Lemma \ref{thm:f_assignments_can_be_extended}, for every falsified ground instance $C\lambda$ of $C$,
we can turn $\alpha$ into a satisfying assignment of $C\lambda$ by flipping the truth value of $\blit\lambda$.
Moreover, this flipping does not falsify any clauses in $F'$.
The only clauses that could possibly be falsified are other 
ground instances of $C$ 
%in $F_C$
 that contain the literal $\nblit\lambda$. 

But, once an instance $\blit\tau$ of the blocking literal $\blit$ is true
in a ground instance $C\tau$ of $C$, this ground instance cannot (later) be falsified by making other instances of
$\blit$ true (since it has, of course,  the same polarity as $\blit$). 
As there are only finitely many clauses in $F_C$, we can therefore
turn $\alpha$ into a satisfying assignment of $F' \cup F_C$ by repeatedly
making ground instances of $C$ true by flipping the truth
values of their instances of the blocking literal $\blit$.
\end{proof}

%\begin{theorem}\label{thm:bc_fo_is_a_redundancy_property}
%If a clause without equality is blocked in a formula $F$ without equality, it is redundant \wrt $F$.
%\end{theorem}
\begin{theorem}\label{thm:bc_fo_is_a_redundancy_property}
If a clause is blocked in a formula $F$, it is redundant \wrt $F$.
\end{theorem}

%\todo{Maybe there is a better formulation of the theorem.}

\begin{proof}
Let $C$ be blocked by $L$ in $F$ and suppose \mbox{$F \setminus \{C\}$} is satisfiable. 
We show that $F \cup \{C\}$ is satisfiable.
By Herbrand's theorem (Theorem~\ref{thm:herbrands_theorem}), it suffices to show that every finite set of ground instances of clauses in $F \cup \{C\}$ is propositionally satisfiable.
Let therefore
$F'$ and $F_C$ be finite sets of ground instances of clauses in $F \setminus \{C\}$ and~$\{C\}$, respectively.
Clearly, $F'$ must be propositionally satisfiable for otherwise $F \setminus \{C\}$ were not satisfiable.
By Lemma~\ref{thm:f_assignments_can_be_extended}, 
every satisfying propositional assignment of $F'$ can be turned into one of $F' \cup F_C$.
%thus $F' \cup F_C$ is propositionally satisfiable.
% and so $F' \cup F_C$ is satisfiable.
It follows that $F \cup \{C\}$ is satisfiable.
\end{proof}

%-----------------------------------------------------------------------------------------------------------
\section{Equality-Blocked Clauses 
%in FOL with Equality
%in \foleq
}
\label{sec:bc_fol_with_equality}
%-----------------------------------------------------------------------------------------------------------

In the following, we first illustrate why the blocking notion from the previous section 
fails to guarantee redundancy in the presence of equality.
We then introduce 
%an adapted 
a refined
notion of blocking, \emph{equality-blocking}, and prove that 
equality-blocked clauses are redundant even if 
%the formula at hand contains the equality predicate.
the equality predicate is present.

\begin{example}\label{ex:blocking_fails_with_equality}
Let $C = P(a)$ and $F = \{a \foeq b, \neg P(b)\}$. 
Since $P(a)$ and $P(b)$ are not unifiable, there are no resolvents of~$C$, 
hence $P(a)$ trivially blocks~$C$ in $F$.
%$P(a)$ trivially blocks~$C$ in $F$ since $P(b)$ is not unifiable with $P(a)$.
But, $F$ is clearly satisfiable whereas $F \cup \{C\}$ is not.
\end{example}

\noindent
In Example \ref{ex:blocking_fails_with_equality}, every model of $F$ must
assign the same truth value to $P(a)$ and $P(b)$.
Hence, when trying to turn a model of $F$ into one of $F \cup \{C\}$ by flipping the truth value
of $P(a)$, we implicitly flip the truth value of $P(b)$ although $P(a)$ and $P(b)$ are not unifiable.

Thus, in the presence of equality, %when testing whether a clause $C$ is blocked by a literal $L(t_1, \dots, t_n) \in C$,
it is not enough to consider only the clauses that are resolvable with~$C$.
We need to take all clauses that contain a literal of the form $\bar L(\dots)$ into account.
In order to do so, we make use of \emph{flattening} as introduced by Khasidashvili and Korovin~\cite{khasidashvili16_fo_predicate_elimination}:

\begin{definition}\label{def:flattening}
Let $C = L(t_1, \dots, t_n) \lor C'$. \emph{Flattening} the literal $L(t_1, \dots, t_n)$ in $C$ yields the clause
$C^- = \bigvee_{1 \leq i \leq n} x_i \foneq t_i \lor L(x_1, \dots, x_n) \lor C'$,
with $x_i, \dots, x_n$ being fresh variables not occurring in $C$.
%We call $L(x_1, \dots, x_n)$ the \emph{flattened literal}. \todo{This might need changes/renaming, whatever.}
\end{definition}

\begin{example}
Flattening the literal $P(f(x),c,c)$ in clause $P(f(x),c,c) \lor Q(c)$ yields the new 
clause
$x_1 \foneq f(x)$ $\lor$ $x_2 \foneq c \lor x_3 \foneq c \lor P(x_1, x_2, x_3) \lor Q(c)$.
\end{example}

\noindent
The clause resulting from flattening $L(t_1, \dots, t_n)$ in $L(t_1, \dots, t_n) \lor C'$ is equivalent to an implication of the form 
\mbox{$(x_1 \foeq t_1 \land \dots \land x_n \foeq t_n) \rightarrow (L(x_1, \dots, x_n) \lor C')$}. 
Thus, flattening preserves equivalence.
Using flattening, we can define \emph{flat resolvents}. Intuitively, flat resolvents are obtained by first
flattening literals and then resolving them. 
%Since flattened literals with the same predicate symbol and polarity
%are trivially unifiable, 
This enables us to resolve literals that might otherwise not be unifiable.

\iffalse
\begin{definition}
Let $C = L_1 \lor \dots \lor L_k \lor C'$ and $D = N_1 \lor \dots \lor N_l \lor D'$ be clauses with the literals $L_1, \dots, L_k, \bar N_1, \dots, \bar N_l$ having the same predicate symbol and polarity. Let furthermore $C_f$ and $D_f$ be obtained from $C$ and $D$, respectively, by flattening the literals $L_1, \dots, N_l$ and denote the flattened literals by $\flatt{L_1}, \dots, \flatt{N_l}$.
The resolvent 
\begin{align*}
	R =\ 	&(C_f \setminus \{\flatt{L_1},\dots,\flatt{L_k}\})\sigma\ \lor\\
		&(D_f \setminus \{\flatt{N_1}, \dots, \flatt{N_l}\})\sigma
\end{align*} 
of $C_f$ and $D_f$, 
with $\sigma$ being an $\mgu$ of $\flatt{L_1}, \dots, \flatt(L_k)$, $\flatt{\bar N_1}, \dots, \flatt{\bar N_l}$, 
is a \emph{flat resolvent} of $C$ and~$D$. If $k = 1$, i.e., if $R$ is a $\flatt{L_1}$-resolvent of $C_f$ and $D_f$, it is a \emph{flat $L_1$-resolvent} of $C$ and~$D$.
\end{definition}
\fi

\begin{definition}
Let $C = L \lor C'$ and \mbox{$D = N_1 \lor \dots \lor N_l \lor D'$} with $l>0$ be clauses such that the literals $L, \bar N_1, \dots, \bar N_l$ 
have the same predicate symbol and polarity. Let furthermore $\flatt{C}$ and~$\flatt{D}$ be obtained from $C$ and $D$, respectively, by flattening $L, N_1, \dots, N_l$ and denote the flattened literals by $\flatt{L}, \flatt{N_1}, \dots, \flatt{N_l}$.
The resolvent 
\begin{align*}
	(\flatt{C} \setminus \{\flatt{L}\})\sigma\ \lor (\flatt{D} \setminus \{\flatt{N_1}, \dots, \flatt{N_l}\})\sigma
\end{align*} 
of $\flatt{C}$ and $\flatt{D}$, 
with $\sigma$ being an $\mgu$ of $\flatt{L}, \flatt{\bar N_1}, \dots, \flatt{\bar N_l}$, 
is a \emph{flat $L$-resolvent} of $C$ and~$D$.
\end{definition}

\noindent
Note that the unifier $\bigcup_{i = 1}^{l} \{y_{ij} \mapsto x_j \mid 1 \leq j \leq n\}$ of $L(x_1, \dots, x_n)$, $\bar N_1(y_{11}, \dots, y_{1n})$, \dots, $\bar N_l(y_{l1}, \dots, y_{ln})$ is a most general unifier
(cf.~%
%Lemma 4.6.3 in~
\cite{baader98_term_rewriting}).
\todo{Maybe we can find a simpler formulation for the above or we even leave it out.}
 
\begin{example}\label{ex:flat_resolvent}
Let $C = P(a)$ and $D = \neg P(b)$ \textup{(}cf.~Example~\ref{ex:blocking_fails_with_equality}\textup{)}.
By flattening $P(a)$ in $C$ and $\neg P(b)$ in $D$ we obtain
$\flatt{C} = x_1 \foneq a \lor P(x_1)$ and $\flatt{D} = y_1 \foneq b \lor \neg P(y_1)$, respectively.
Their 
\iftautology
\textup{(}non-tautological\textup{)} 
\fi
resolvent
$x_1 \foneq a \lor x_1 \foneq b$
\iftautology\else
\textup{(}which is not valid\textup{)} 
\fi
is a flat \mbox{$P(a)$-resolvent} of $C$ and $D$.
\end{example}

%\begin{example}
%Let $C = P(a)$ and $D = \neg P(b)$. % \textup{(}cf.~Example~\ref{ex:blocking_fails_with_equality}\textup{)}.
%Flattening $P(a)$ in $C$ and $\neg P(b)$ in $D$ yields
%$x_1 \foneq a \lor P(x_1)$ and $y_1 \foneq b \lor \neg P(y_1)$, respectively.
%Their \textup{(}non-tautological\textup{)} resolvent $x_1 \foneq a \lor x_1 \foneq b$ is a flat \mbox{$P(a)$-resolvent} of $C$ and $D$.
%\end{example}

\noindent
The following definition prohibits blocking by an equality literal.
This is because equality must be treated specially
in our extension of the flipping argument (see below).
After this intuitive discussion, we formally define equality-blocking as follows:
% After this discussion, we define equality-blocking:

\begin{definition}\label{def:fo_bc_noneq_binary_resolution}
A clause $C$ is \emph{equality-blocked} by a literal $L \in C$ in a formula $F$ if 
the predicate of $L$ is not $\foeq$ and
all flat $L$-resolvents of $C$ %\textup{(}
with clauses in $F \setminus \{C\}$ %\textup{)} 
are 
\iftautology
tautologies.
\else
valid.
\fi
\end{definition}

\noindent
\iftautology\else
Note that in the presence of equality, clauses without complementary literals, like $x \foeq x$, can be valid.
\fi
%Note also that this definition prohibts blocking by an equality literal. 
%The need for the special treatment of equality becomes clear in the extension 
%of the flipping argument given below. 
Before we prove redundancy, we consider the 
following example that stems from 
a first-order encoding of an
AI-benchmark problem known as ``Who killed Aunt Agatha?''~\cite{pelletier86_75problems}
and that illustrates the power of equality-blocked~clauses:

\begin{example}\label{ex:eq_blocking_lives}
Let $F$ be the following set of four clauses: 
$\{L(a), L(b), L(c), \neg L(x) \lor x \foeq a \lor x \foeq b \lor x \foeq c\}$.
Intuitively, the clauses $L(a)$, $L(b)$, and $L(c)$ encode that there are three living
individuals: Agatha, Butler, and Charles. 
The clause $\neg L(x) \lor x \foeq a \lor x \foeq b \lor x \foeq c$ encodes that
these three individuals are the only living individuals.
We can observe that all four clauses are equality-blocked in $F$.
For instance, let $C = L(a)$.
There exists one flat $L(a)$-resolvent of~$C$: the 
\iftautology
tautology
\else
valid clause
\fi
$x_1 \foneq a \lor x_1 \foneq x \lor x \foeq a \lor x \foeq b \lor x \foeq c$,
obtained by resolving the clause $x_1 \foneq a \lor L(x_1)$ with 
$y_1 \foneq x \lor \neg L(y_1) \lor x \foeq a \lor x \foeq b \lor x \foeq c$.
%using the $\mgu$ $\{y_1 \mapsto x_1\}$ of $L(x_1)$ and~$L(y_1)$.
\end{example}

\noindent
In order to show that equality-blocked clauses are redundant, we introduce the notion of \emph{equivalence flipping}. 
%Intuitively, a propositional assignment $\alpha'$ is obtained from an assignment $\alpha$ 
%by equivalence flipping the truth value of a ground literal $L(t_1, \dots, t_n)$, by inverting 
Intuitively, equivalence flipping of a ground literal $L(t_1, \dots, t_n)$ turns a propositional assignment $\alpha$ into
an assignment $\alpha'$ by inverting the truth value of $L(t_1, \dots, t_n)$ 
as well as that of all $L(s_1, \dots, s_n)$ for which $\alpha$ satisfies \mbox{$t_1 \foeq s_1, \dots, t_n \foeq s_n$}. 

\begin{definition}
Let $\alpha$ be a propositional assignment and $L(t_1, \dots, t_n)$ a ground literal with predicate symbol~$P$ other than $\foeq$. 
The assignment $\alpha'$, obtained by \emph{equivalence flipping} the truth value of $L(t_1, \dots, t_n)$, is defined as follows:
\begin{align*}
	\alpha'(A) =
	\begin{cases}
		1 - \alpha(A) 	& \text{if $A = P(s_1, \dots, s_n)$ and}\\
					& \text{\hphantom{if }$\alpha(t_i \foeq s_i) = \true$ for all $1 \leq i \leq n$,}\\
		\alpha(A)		& \text{otherwise.}
	\end{cases}
\end{align*}
\end{definition}

\noindent
Obviously,  equivalence flipping preserves the truth of instances of the equality axioms,
\iffalse
Only instances of equality axiom E3, which are of the form 
$t_1 \foneq s_1 \lor \dots \lor t_n \foneq s_n \lor \neg P(t_1,\dots,t_n) \lor P(s_1,\dots, s_n)$, 
could possibly be affected. 
\todo{The next sentence needs rewriting.}
But, whenever $t_1 \foeq s_1, \dots, t_n \foeq s_n$ and $P(t_1,\dots,t_n)$ are satisfied by $\alpha'$, 
$P(s_1,\dots,s_n)$ is also satisfied by $\alpha'$.
\fi
%The following lemma is the 
leading to the equality counterpart of Lemma \ref{thm:fo_flipping_works}:

%\begin{lemma}\label{thm:foeq_flipping_works}
%Let $C$ be equality-blocked by $L$ in $F$, 
%and $\alpha$ a propositional assignment that satisfies all ground instances of the equality axioms, $\eqaxioms$, 
%but falsifies a ground instance $C\lambda$ of $C$. 
%Then, the assignment $\alpha'$, obtained from $\alpha$ by equivalence flipping the truth value of $L\lambda$, 
%still satisfies all the ground instances of clauses in $\eqaxioms \cup F \setminus \{C\}$ that have been satisfied by $\alpha$.
%\end{lemma}
\begin{lemma}\label{thm:foeq_flipping_works}
Let $C$ be equality-blocked by $L$ in $F$, 
and $\alpha$ a propositional assignment that satisfies all ground instances of the equality axioms, $\eqaxioms$, 
but falsifies a ground instance $C\lambda$ of $C$. 
Then, the assignment $\alpha'$, obtained from $\alpha$ by equivalence flipping the truth value of $L\lambda$, 
satisfies all the ground instances of clauses in $\eqaxioms \cup F \setminus \{C\}$ that are satisfied by $\alpha$.
\end{lemma}

\begin{proof}
Let $L = L(t_1, \dots, t_n)$ and $C = L \lor C'$ and suppose $\alpha$ falsifies a ground instance $C\lambda$ of~$C$. 
By definition, the only clauses that are affected by the equivalence flipping of $L(t_1, \dots, t_n)\lambda$ 
are clauses of the form $D\tau$, 
with \mbox{$D \in F \setminus \{C\}$} and $\bar L(s_1, \dots, s_n)\tau \in D\tau$ 
such that \mbox{$\alpha(t_i\lambda \foeq s_i\tau) = \true$} for $1 \leq i \leq n$. 

Let $D\tau$ be such a clause
and let $\bar L(s_1, \dots, s_n)$,\hspace{1.4pt}$\dots$, $\bar L(r_1, \dots, r_n)$ be all literals in $D$ such that
$\alpha$ satisfies \mbox{$t_i\lambda \foeq s_i\tau$},$\dots, t_i\lambda \foeq r_i\tau$ for $1 \leq i \leq n$.
To simplify the presentation, we assume that $\bar L(s_1, \dots, s_n)$ and $\bar L(r_1, \dots, r_n)$ are all such literals.
The proof for another number of such literals is analogous.
We observe that $D$ is of the form $\bar L(s_1, \dots, s_n) \lor \bar L(r_1, \dots, r_n) \lor D'$.

%$\flatt(L(t_1, \dots, t_n)) = L(x_1, \dots, x_n)$, 
%$\flatt(\bar L(s_1, \dots, s_n)) = \bar L(y_1, \dots, y_n)$, and 
%$\flatt(\bar L(r_1, \dots, r_n)) = \bar L(z_1, \dots, z_n)$ 
Since $C$ is equality-blocked by $L(t_1, \dots, t_n)$ in~$F$, 
all flat $L(t_1, \dots, t_n)$-resolvents of $C$ are 
\iftautology
tautologies.
\else
valid.
\fi
Therefore, the flat $L(t_1, \dots, t_n)$-resolvent
%the flat $L(t_1, \dots, t_n)$-resolvent
\begin{align*}
	R = 
	(
%	x_1 \foneq t_1 \lor y_1 \foneq s_1 \lor z_1 \foneq r_1 \lor \dots \lor x_n \foneq t_n \lor y_n \foneq s_n \lor z_n \foneq r_n \lor C' \lor D'
	C' \lor D' \lor \bigvee_{1 \leq i \leq n} x_i \foneq t_i \lor y_i \foneq s_i \lor z_i \foneq r_i
	)\sigma
	%\text{,} 
\end{align*}
is 
\iftautology
a tautology,
\else
valid,
\fi
where $\sigma$ is an $\mgu$ of the literals 
$L(x_1, \dots, x_n), L(y_1, \dots, y_n)$, and $L(z_1, \dots, z_n)$,
which were obtained by respectively flattening $L(t_1, \dots, t_n)$, $\bar L(s_1, \dots, s_n)$, and $\bar L(r_1, \dots, r_n)$.
Assume w.l.o.g.~that 
$\sigma = \{y_i \mapsto x_i \mid 1 \leq i \leq n\} \cup  \{z_i \mapsto x_i \mid 1 \leq i \leq n\}$. 
Then,
\begin{align*}
	R = 
%	x_1 \foneq t_1 \lor x_1 \foneq s_1 \lor x_1 \foneq r_1 \lor \dots \lor x_n \foneq t_n \lor x_n \foneq s_n \lor x_n \foneq r_n \lor C' \lor D'
	C' \lor D' \lor \bigvee_{1 \leq i \leq n} x_i \foneq t_i \lor x_i \foneq s_i \lor x_i \foneq r_i
	\text{.} 
\end{align*}
As $R$ is 
\iftautology
a tautology,
\else
valid,
\fi
the assignment $\alpha$ must satisfy all ground instances of~$R$.
Consider therefore the following substitution $\gamma$ that yields a ground instance $R\gamma$ of $R$:
\begin{align*}
\gamma(x) =
	\begin{cases}
		t_i\lambda	& \text{if $x \in \{x_1, \dots, x_n\}$,}\\
		x\lambda		& \text{if $x \in \var(C')$,}\\
		x\tau 		& \text{if $x \in \var(D')$.}
	\end{cases}	
\end{align*}
We observe that the ground instance $R\gamma$ of $R$ is the clause
\begin{align*}
	%R\gamma =\ 	
%	&t_1\lambda \foneq t_1\lambda \lor \dots \lor t_n\lambda \foneq t_n\lambda \lor \\
%	&t_1\lambda \foneq s_1\tau \lor \dots \lor t_n\lambda \foneq s_n\tau \lor \\
%	&t_1\lambda \foneq r_1\tau \lor \dots \lor t_n\lambda \foneq r_n\tau \lor \\
	&C'\lambda \lor D'\tau  \lor 
	\bigvee_{1 \leq i \leq n} t_i\lambda \foneq t_i\lambda \lor t_i\lambda \foneq s_i\tau \lor t_i\lambda \foneq r_i\tau %\text{.}
\end{align*}
which must be satisfied by $\alpha$.
Now, all the $t_i\lambda \foneq t_i\lambda$ are clearly falsified by $\alpha$. 
Furthermore, by assumption, $\alpha$ falsifies all the $t_i\lambda \foneq s_i\tau$ 
and all the $t_i\lambda \foneq r_i\tau$ as well as $C'\lambda$.
But then, $\alpha$ must satisfy at least one of the literals in $D'\tau$.
Since %, by assumption, 
none of the literals in $D'\tau$ are affected by equivalence flipping the truth value of $L\lambda$, 
$D'\tau$ must be satisfied by $\alpha'$.
It follows that $\alpha'$ satisfies $D\tau$.
\end{proof}

\noindent
Using Lemma \ref{thm:foeq_flipping_works} instead of Lemma \ref{thm:fo_flipping_works} 
and replacing the notion of flipping by that of equivalence flipping, 
the proof of the following lemma is analogous to the one of Lemma \ref{thm:f_assignments_can_be_extended}:

\begin{lemma}\label{thm:feq_assignments_can_be_extended}
Let $C$ be a clause that is equality-blocked in $F$. 
Let furthermore $F'$ and $F_C$ be finite sets of ground instances of clauses in $F \setminus \{C\}$ and $\{C\}$, respectively. 
Then, every assignment that propositionally satisfies all the ground instances of  $\eqaxioms$ as well as $F'$ 
can be turned into one that satisfies $F' \cup F_C$ and all the ground instances of $\eqaxioms$.
\end{lemma}

\noindent
Using Lemma \ref{thm:feq_assignments_can_be_extended} and the equality variant of Herbrand's Theorem (Theorem~\ref{thm:herbrand_with_equality}), the proof of the following theorem is similar to that of Theorem~\ref{thm:bc_fo_is_a_redundancy_property}:

\begin{theorem}\label{thm:equalityblocked_redundant}
If a clause is equality-blocked in a formula~$F$, it is redundant \wrt $F$.
\end{theorem}

\begin{proof}
%The idea of the proof is similar to that of Theorem~\ref{thm:bc_fo_is_a_redundancy_property}.
Let $C$ be equality-blocked in $F$.
Assuming that $F \setminus \{C\}$ is satisfiable, we conclude that there exists an assignment $\alpha$
that propositionally satisfies all ground instances of clauses in $(F \setminus \{C\}) \cup \eqaxioms$ 
(by using Theorem~\ref{thm:herbrand_with_equality} together with compactness).
\todo{Should we state the compactness theorem?}
Using Lemma~\ref{thm:feq_assignments_can_be_extended}, one can then show that
every finite set of ground instances of clauses in $F \cup \{C\} \cup \eqaxioms$ can be satisfied 
by modifying~$\alpha$.
It follows, again by Theorem~\ref{thm:herbrand_with_equality}, that $F \cup \{C\}$ is satisfiable.
\end{proof}

%-----------------------------------------------------------------------------------------------------------
\section{Complexity of Detecting Blocked Clauses}
\label{sec:complexity}
%-----------------------------------------------------------------------------------------------------------

%\begin{algorithm}
%\caption{Testing whether a clause is blocked}
%\label{alg:decide_blockedness}
%\begin{algorithmic}[1]
%\Require
%\Statex A clause $C$, a literal $L \in C$, and a formula $F$
%\Ensure
%\Statex Indication whether $C$ is blocked by $L$ in $F$
%\Statex
%\For{each clause $D \in F \setminus \{C\}$}
%	\If{there exists an invalid $L$-resolvent of $C$ and $D$}
%		\State	\Return	NO
%	\EndIf
%\EndFor
%\State	\Return	YES
%\end{algorithmic}
%\end{algorithm}

%\begin{algorithm}
%\caption{Testing whether a clause is blocked}
%\label{alg:decide_blockedness}
%\begin{algorithmic}[1]
%\Require
%\Statex A clause $C$ and a formula $F$
%\Ensure
%\Statex Indication whether $C$ is blocked in $F$
%\Statex
%\For{each literal $L \in C$}  \label{alg_lres_valid:main_loop}
%	\State $allvalid \gets YES$
%	\For{each clause $D \in F \setminus \{C\}$}
%		\If{there exists an invalid $L$-resolvent of $C$ and $D$}
%			\State $allvalid \gets NO$
%		\EndIf
%	\EndFor
%	\If{$allvalid$ = YES}
%		\State	\Return	YES
%	\EndIf
%\EndFor
%\State	\Return	NO
%\end{algorithmic}
%\end{algorithm}

In this section, we show that deciding whether a clause $C$ is 
blocked (or equality-blocked) by a literal $L$ in a formula $F$
can be decided in polynomial time.
From the definitions of blocking (Definition~\ref{def:fo_bc_noneq}) 
and equality-blocking (Definition \ref{def:fo_bc_noneq_binary_resolution}) 
this is not obvious, because a direct implementation of these definitions 
would require to test exponentially 
many (flat) $L$-resolvents of $C$ for validity.
Although the number of clauses in $F \setminus \{C\}$ with which $C$ could 
possibly be resolved is linearly bounded by the size of $F$,  
there can be exponentially many $L$-resolvents of $C$ with a single 
clause $D \in F \setminus \{C\}$.
For example, consider the clause $C = L \lor C'$ and 
assume that $F \setminus \{C\}$ contains a clause 
$D = N_1 \lor \dots \lor N_n \lor D'$ such that the literals 
$L,\bar N_1, \dots, \bar N_n$ are unifiable. 
We then have one $L$-resolvent of $C$ and $D$ for every non-empty subset of 
$N_1, \dots, N_n$. 
Therefore, there are $2^n - 1$ such $L$-resolvents whose validity we have to check.
To show how this can be done in polynomial time, we first argue that the validity of a (flat) 
$L$-resolvent is decidable in polynomial time and then show that it actually suffices to check the validity of
only polynomially many $L$-resolvents.

%The key to a polynomial time solution lies
%in realizing that if a certain $L$-resolvent
%is valid many other $L$-resolvents must be valid as well. This idea, which is common to the dealing with
%both the equational and non-equational case, will be turned into an algorithm later in this section. 
%We start by discussing those details which need to be dealt with separately.

In the case without equality, checking the validity of an 
$L$-resolvent basically amounts to looking 
for a complementary pair of literals. Assume we want to check 
the validity of an $L$-resolvent 
$(C' \lor D')\sigma$ of clauses $C' \lor L$ and $N_1 \lor \dots \lor N_n \lor D'$ where $\sigma$ is an $\mgu$
of $L,\bar N_1, \dots, \bar N_n$. Although the size of $\sigma$ can be exponential in the worst case,
this exponential blow up can be avoided by not computing $\sigma$ explicitly but only computing the \emph{unification closure} \cite{kanellakis89_ccanduc} of $L,\bar N_1, \dots, \bar N_n$, which can be done in polynomial time.
The unification closure is basically an equivalence relation under which two literals are considered equivalent if
they are unified by a most general unifier of $L,\bar N_1, \dots, \bar N_n$.
Checking the validity of $(C' \lor D')\sigma$ then boils down to checking whether $C' \lor D'$ contains two literals 
that are complementary \wrt the unification closure.
\todo{Maybe add that it is perfectly natural to consider blocking only in the case where no equality
is present, because otherwise blocking does anyhow not guarantee redundancy.}

In the case of equality-blocking, we work with flat $L$-resolvents.
Computing a flat $L$-resolvent is easy
since---as pointed out in the section on equality-blocked clauses---there 
exists a trivial (and small) $\mgu$ of the flattened literals.
Furthermore, a flat $L$-resolvent $R$ is valid
if and only if the negation of its universal closure $\neg \forall R$
is unsatisfiable. After skolemization (which introduces fresh constants for the variables of $R$),
the formula $\neg \forall R$ becomes a conjunction of ground (equational) literals
and can therefore be efficiently decided by a congruence-closure algorithm~(cf.~\cite{shostak78_equality_reasoning}).

\todo{Decide for below and use consistently: Algorithm \ref{alg_lres_valid}, the algorithm, or procedure?}
%We now move on to presenting a 
 Algorithm~\ref{alg_lres_valid} shows a
polynomial-time procedure for checking
whether all $L$-resolvents of a candidate clause $C = L \lor C'$ and 
a partner clause $D$ are valid.
We do this here for the non-equational case and leave the details of the equational case for the appendix (see Appendix~\ref{sec_equational_cplx}).
%It should be clear that 
With  this procedure 
 deciding whether a clause $C$ is blocked in a formula $F$
can be done in polynomial time by iterating over all the potential blocking literals $L \in C$ 
and all the partner clauses $D \in F \setminus \{C\}$. Practical details on how to efficiently 
implement this top-level iteration will be discussed in Section~\ref{sec:implementation}.

\begin{algorithm}[t]
\caption{Testing validity of $L$-resolvents}
\label{alg_lres_valid}
\begin{algorithmic}[1]
\Require
\Statex A candidate clause $C = L \lor C'$ and a partner clause $D = N_1 \lor \ldots \lor N_n \lor D'$,
\Statex where the literals $N_1, \ldots, N_n$ are all the literals of $D$ which pairwise unify with $\bar{L}$
\Ensure
\Statex Indication whether all $L$-resolvents of $L \lor C'$ and $D$ are valid
\Statex
\For{$k \gets 1,\ldots,n$}  \label{alg_lres_valid:main_loop}
\State $N \gets \{N_k\}$
	\While{$L$ is unifiable with the literals in $\bar N$ via an $\mgu$ $\sigma$} \label{alg_lres_valid:while_loop}	 
	\State Let $K$ contain all pairs of complementary literals in the $L$-resolvent $C'\sigma \lor (D \setminus N)\sigma$
	\If{$K = \emptyset$} \label{alg_lres_valid:checkvalid} %\hspace{20pt}($\rightarrow$ the resolvent is not valid)\label{alg_lres_valid:checkvalid}
		\State	\Return	NO  \label{alg_lres_valid:fail}	
	\EndIf
	\If{every pair of complementary literals in $K$ contains a literal $N_i\sigma$} \label{alg_lres_valid:branch}
		\State $N \gets N \cup \{N_i \mid \text{$N_i\sigma$ is part of a complementary pair}\}$ \label{alg_lres_valid:branch_yes}
		%\hspace{5pt} (for some $i$ such that $N_i\sigma$ is part of a complementary pair)
		%\label{alg_lres_valid:branch_yes}
	\Else
		\State \textbf{break} (the while loop) \label{alg_lres_valid:while_loop_end}
	\EndIf
	\EndWhile
\EndFor
\State	\Return	YES
\end{algorithmic}
\end{algorithm}

The inputs of Algorithm~\ref{alg_lres_valid} are 
a candidate clause $C = L \lor C'$ and a partner clause $D = N_1 \lor \ldots \lor N_n \lor D'$,
where the literals $N_1, \ldots, N _n$ are all the literals of $D$ which pairwise unify with $\bar{L}$.
%In the algorithm, the set $I$ is used to store a subset of the indexes $\{1,\ldots,n\}$ whereby
%every such subset corresponds to an $L$-resolvent $C'\sigma \lor (D \setminus \{M_i \ |\ i \in I\})\sigma$ with $\sigma$
%being an $\mgu$ of $\bar L$ and $\{M_i \ |\ i \in I\}$. 
It is easy to see that the running time of the procedure is quadratic in~$n$:
We perform $n$ iterations of the for loop and at most $n$ iterations of the inner while loop (since there are no more than $n-1$ literals $N_i$ that can be added to $N$ in line~\ref{alg_lres_valid:branch_yes}).
Therefore, only quadratically many $L$-resolvents are explicitly tested for validity. 
By proving that Algorithm~\ref{alg_lres_valid} is a sound and complete procedure for testing whether
all $L$-resolvents of a candidate clause $C$ with a partner clause $D$ are valid, we show that this is sufficient:

\begin{theorem}\label{thm:alg_soundness}
Algorithm \ref{alg_lres_valid} returns YES if and only if all $L$-resolvents of $C$ with $D$ are valid.
\end{theorem}

\begin{proof}
For the $\Leftarrow$-direction, assume that all $L$-resolvents of $C$ with $D$ are valid, i.e., 
they all contain at least one pair of complementary literals. 
It follows that line~\ref{alg_lres_valid:fail} is never executed and therefore the algorithm returns YES.

For the $\Rightarrow$-direction, let $C = C' \lor L$ and let $R = C'\sigma \lor (D \setminus N)\sigma$ 
be an $L$-resolvent of $C$ and $D$ with $\sigma$
being an $\mgu$ of $L$ and $\bar N$ (note that $N$ is a \emph{set} of literals).
If the algorithm has explicitly tested $R$ for validity (in line~\ref{alg_lres_valid:checkvalid}), then the statement clearly holds.
Assume thus that $R$ has not been explicitly tested for validity.
Now, let $N'$ be a maximal subset of $N$ for which the validity of the $L$-resolvent $R' = C'\sigma' \lor (D \setminus N')\sigma'$
(with $\sigma'$ being an $\mgu$ of $L$ and $\bar N'$) has been explicitly tested.\footnote{In other words, 
$N'$ should be a subset of $N$ for which there exists no other subset $N''$ of $N$ such that
\begin{inparaenum}[(1)]
	\item $N' \subset N''$, and
	\item the validity of an $L$-resolvent $C'\sigma'' \lor (D \setminus N'')\sigma''$ has been explicitly tested by the algorithm in line~5.
\end{inparaenum}
}
Clearly, such an $N'$ must exist since the algorithm explicitly tests the validity of all binary resolvents upon $L$ 
(in the first iteration of the while-loop, for every iteration of the for-loop).
As the algorithm returned YES, we know that $R'$ must be valid.

From $N'$ being maximal it follows that $R'$ contains a complementary pair of literals 
$P(t_1, \dots, t_n)\sigma'$ and $\neg P(s_1, \dots, s_n)\sigma'$ such that $P(t_1, \dots, t_n)$ and $\neg P(s_1, \dots, s_n)$ are both not contained in $N$: 
 otherwise the algorithm would have continued by testing the validity of
an $L$-resolvent $C'\sigma'' \lor (D \setminus N'')\sigma''$  with $N' \subset N'' \subseteq N$ 
(by extending $N'$ in line~\ref{alg_lres_valid:branch_yes} and then testing validity in the next iteration of the while-loop).
It follows that $P(t_1, \dots, t_n)\sigma$ and $\neg P(s_1, \dots, s_n)\sigma$ are both contained in~$R$.
Now, since $\sigma$ unifies $L$ with $\bar N$, it unifies $L$ with $\bar N'$. Moreover, since $\sigma'$ is a most general unifier of $L$ and $\bar N'$,
it is more general than $\sigma$ and therefore
there exists a substitution $\gamma$ such that $\sigma'\gamma = \sigma$.
But then, since $P(t_1, \dots, t_n)\sigma' = P(s_1, \dots, s_n)\sigma'$
it follows that $P(t_1, \dots, t_n)\sigma = P(t_1, \dots, t_n)\sigma'\gamma = P(s_1, \dots, s_n)\sigma'\gamma = P(s_1, \dots, s_n)\sigma$
and thus $R$ is valid.
\end{proof}

%In every iteration of the for-loop, the algorithm starts
%by testing the validity of a binary resolvent and the
%---depending on whether it is valid and, in case it is not, whether all the complementary pairs contain literals $M_i\sigma$ with $i \notin I$---
\noindent
In conclusion, we have shown that it suffices to perform polynomially many 
validity checks---each of which can be performed in polynomial time---to decide whether a clause
is blocked.

\section{Blocked-Clause Elimination in First-Order Logic}
\label{sec:bce}
%-----------------------------------------------------------------------------------------------------------

In this section, we present the implementation and empirical evaluation of a first-order preprocessing tool that
performs one possible application of blocked clauses, namely blocked-clause elimination~(BCE). We further discuss how
BCE eliminates pure predicates and how it is related to the existing preprocessing technique of \emph{unused definition
elimination} (UDE) by Hoder et al.~\cite{DBLP:conf/fmcad/HoderKKV12}.

\iffalse
Write a transition paragraph. 
Mind that:
\begin{itemize}
\item
Until now blocked clauses were an independent concept with potential manifold applications. From now on, we focus on elimination only!

(Re-)introduce the acronym BCE! 

\item
Also we transition from pure theory to practice / implementation / experiments.
\item
Already pre-explain why we have a subsection on ``relation to UDE'' -- disclaimer style (?)
\end{itemize}
\fi

\subsection{Implementation}
\label{sec:implementation}

We implemented blocked-clause elimination and equality-blocked-clause elimination for first-order logic 
as a preprocessing step in the automated theorem prover \vampire{} \cite{KovacsVoronkov:CAV:Vampire:2013}.\footnote{
A statically compiled x86\_64 executable of \vampire{} used in our experiments can be obtained from 
\url{http://forsyte.at/wp-content/uploads/vampire_bce.zip}.}
This preprocessing step can be activated by providing the command line flag \texttt{-bce on}.
Depending on whether the formula at hand contains the equality predicate or not, \vampire{} then
performs either the elimination of equality-blocked clauses or blocked clauses.
It will be performed as the last step in the preprocessing pipeline, because it relies on the input being in CNF. 
After the preprocessing, instead of proceeding to proving the formula---which is the default behavior---\vampire{} 
can be instructed to output the final set of clauses by specifying \texttt{--mode clausify} on the command line.

%Here we briefly describe the implementation of % first-order 
%blocked clause elimination as used in our experiments. 
%The overall elimination procedure is organized as a fixed-point loop
%with ``touched literals'' similarly to the one described for the propositional case \cite{BlockedClauseElimination}.
%%
%The unit processed in each iteration is a clause $C$ together with its literal $L$, a candidate for the blocking literal.
%The main feature of the loop is that when a test for blocking of a particular pair $(C,L)$
%fails because of another clause $D$, the testing for $(C,L)$ gets resumed when $D$ is itself shown blocked.

The top level organization of our elimination procedure, which is the same for both blocked-clause elimination and equality-blocked-clause elimination, is inspired by the approach adopted
in the propositional case by J\"arvisalo et al.\ (c.f.\ \cite{BlockedClauseElimination}, section 7).
For efficiency, we maintain an index for accessing a literal within a clause by its predicate symbol and polarity.
The main data structure is a priority queue of candidates $(L,C)$ where $L$ is a potential blocking literal in a clause $C$.
We prioritize for processing those candidates $(L,C)$ which have fewer potential resolution partners 
estimated by the number of clauses indexed with the same predicate symbol and the opposite polarity as $L$.\footnote{
We remark that, similarly to the propositional case, blocked-clause elimination in first-order logic is confluent.
This means that the resulting set of clauses is always the same regardless of the elimination order.
The ordering of candidates in our queue can therefore influence the computation time,
but not the output of our procedure.}

At the beginning, every (non-equational) literal $L$ in a clause $C$ gives rise to a candidate $(L,C)$.
We always pick the next candidate $(L,C)$ from the queue and iterate over potential resolution partners~$D$.
%We make sure to skip clause $C$ itself and every partner clause $D$ which was already shown to be blocked.
If we discover that a (flat) $L$-resolvent of $C$ and $D$ is not valid, further processing of $(L,C)$ 
is postponed and the candidate is ``remembered'' by the partner clause $D$. If, on the other hand, 
all the (flat) $L$-resolvents with all the possible partners $D$ have been found valid, the clause $C$ is declared blocked
and the candidates remembered by $C$ are ``resurrected'' and put back to the queue.
Their processing will be resumed by iterating over those partners which have not been tried yet.

%Depending on the absence or presence of the equality predicate in the problem, 
%it is globally decided before the main loop starts, whether a test for blocking or equality-blocking, respectively,
%will be employed between each candidate and its partner clauses.
Although, as we have shown in Section~\ref{sec:complexity}, testing whether
all the (flat) $L$-resolvents of a clause $C$ and a partner clause $D$ are valid can
be done in polynomial time, our implementation uses for efficiency reasons
an approximate solution, which only computes binary (flat) resolvents.
Then, before testing the resolvent for validity, we remove from it all the literals that
\begin{inparaenum}[(1)]
	\item are unifiable with $\bar L\sigma$ in the blocking case, or
	\item have the same predicate symbol and polarity as $\bar L$ in the equality-blocking case.
\end{inparaenum}
This still ensures redundancy and significantly improves the performance.

For testing validity of flat $L$-resolvents in the equality case,
we experimented with a complete congruence-closure procedure 
which turned out to be too inefficient. Our current implementation
only ``normalizes'' in a single pass all (sub-)terms of the literals in the flat resolvent
using the equations from the flattening,
but ignores (dis-)equations originally present in the two clauses
and does not employ the congruence rule recursively.
% Intuition: This should be enough for doing as much as possible in the case of definitions
% where the defined predicate is flat and linear. Then its variable arguments can be rewritten everywhere to the matching terms from the partner literal.
Our experiments show that even this limited version is effective.

\iffalse

ResolvesToTautology $(C, L_C, D, L_D)$ Non-eq:
- unify via $\sigma$
- apply $\sigma$ to lits of C (checking if C a tautology itself); skip $L_C$; collect remaining lits 
- apply $\sigma$ to lits D (checking if D a tautology itself); skip $L_D$; 
	collect remaining lits; if we find a complementary pair, this new literal mustn't unify with $L_C\sigma$

ResolvesToTautology $(C, L_C, D, L_D)$ Eq:
- let $L_C = p(s_1,\ldots,s_n)$ and $L_D = \neg p(t_1,\ldots,t_n)$
- UnionFind over $1,\ldots, 2n$ 
	- identifies $i$ and $j$ if $s_i = s_j$ (indexes permitting)
	- identifies $i$ and $j$ if $t_{i-n} = s_{j-n}$ (indexes permitting)
	- identifies $i$ and $i+n$ (indexes permitting)

- rewrite all terms occurring as subterms in literals of $C$ that are equal to some $s_1,\ldots,s_n$ or some of the ground ones among $t_1,\ldots,t_n$ to whatever their equivalence class dictates (it's one of $s_1,\ldots,s_n$)
- test if $s = s$ has been obtained

- rewrite all terms occurring as subterms in literals of $D$ that are equal to some $t_1,\ldots,t_n$ or some of the ground ones among $s_1,\ldots,s_n$ to whatever their equivalence class dictates (it's one of $s_1,\ldots,s_n$)
- at the same time rename all other variables of $D$ apart from any of $C$
- test if $s = s$ has been obtained
- test if there is a complementary one to a one of those from previous step. However, it needs to be of a different predicate or polarity from $L_C$

\fi

%-----------------------------------------------------------------------------------------------------------
\subsection{Relation to Pure Predicate Elimination and Unused Definition Elimination}
\label{sec:related}
%-----------------------------------------------------------------------------------------------------------

\iffalse
In this section, we shortly discuss how blocked-clause elimination performs the elimination of \emph{pure predicates} 
and how it is related to the technique of \emph{unused definition elimination} (UDE) 
by Hoder et al.~\cite{DBLP:conf/fmcad/HoderKKV12}.
\fi

% Similarly, the combination of predicate elimination by Khasidashvili and Korovin \cite{khasidashvili16_fo_predicate_elimination} and BCE would be interesting in first-order logic. 

% We leave it for future work to explore this analogy in more detail.

In the propositional setting, blocked-clause elimination is known to simulate on the CNF-level
several refinements of the standard CNF encoding for circuits %, without explicit knowledge of the underlying circuit structure 
\cite{DBLP:journals/jar/JarvisaloBH12}. Somewhat analogously, 
we observe that in the first-order setting BCE simulates 
\emph{pure predicate elimination} (\ppe{}) and, under certain conditions, also
\emph{unused definition elimination} (\ude{}),
a formula-level simplification described by Hoder et al.\ \cite{DBLP:conf/fmcad/HoderKKV12}.
This section briefly recalls these two techniques and explains their relation to BCE. 
% Footnote about more stuff in an appendix?
Apart from being of independent interest,
the observations made in this section 
are also relevant for interpreting the experimental results presented in Section~\ref{sec:experiments}.

% Both \ppe{} and \ude{} are implemented in \vampire{} and enabled by default.

We say that a predicate symbol $P$ is \emph{pure} in a formula $F$ if, in $F$, 
all occurrences of literals with predicate symbol $P$ are of the same polarity. 
If a clause $C$ contains a literal $L$ with a pure predicate symbol $P$, 
then there are no $L$-resolvents of $C$, hence it is vacuously blocked. 
Therefore, blocked-clause elimination removes all clauses that contain pure predicates
and thus simulates \ppe{}.
%\footnote{
%Pure predicate elimination was first introduced
%under the name \emph{affirmative-negative rule} by Davis and Putnam \cite{DBLP:journals/jacm/DavisP60}.}

UDE is a preprocessing method that removes so-called unused predicate definitions from general formulas 
(i.e., formulas that are not necessarily in CNF).
Given a predicate symbol $P$ and a general formula $\varphi$ such that $P$ does not occur in $\varphi$,
a \emph{predicate definition} is a formula
\[\mathit{def}(P,\varphi) = 
\forall \vec{x}.\ P(\vec{x}) \leftrightarrow \varphi(\vec{x}).\]
Assuming we have a predicate definition % $\mathit{def}(p,\varphi)$ 
as a conjunct within a larger formula 
% \begin{equation} \label{formula_Psi}
$\Psi = \psi \land \mathit{def}(P,\varphi),$
% \end{equation}
the definition is \emph{unused} if $P$ does not occur in $\psi$.
(In fact, if $P$ only occurs in $\psi$ with a single polarity,
then one of the two implications of the equivalence $\mathit{def}(P,\varphi)$, corresponding to that polarity,
can be dropped by UDE.)
UDE preserves satisfiability equivalence~\cite{DBLP:conf/fmcad/HoderKKV12}.

%To discuss whether BCE can simulate UDE, we have to take into account 
Note that UDE operates on the level of 
general formulas %---eliminating predicate definitions in the form of universally quantified equivalences---
while BCE is only defined for formulas in CNF. 
Let therefore $\mathit{def}(P,\varphi)$ be an unused predicate definition in the formula $\Psi = \psi \land \mathit{def}(P,\varphi)$ as above  %$\mathit{UDE}(\Psi)$ be the result of removing all unused predicate definitions from a general formula $\Psi$, 
and let $\mathit{BCE}(\cnf{}(\Psi))$ be the result of eliminating all blocked clauses from a clause form translation $\cnf(\Psi)$ of $\Psi$.
%Given the outcome of our experiments together with some intuitive observations, 
We conjecture that for any ``reasonably behaved'' clausification procedure $\cnf$ (e.g., the well-known Tseitin encoding~\cite{tseitin_enc}), it holds that
$
\mathit{BCE}(\cnf{}(\Psi)) \subseteq \cnf(\psi)%\text{,}
$
if $\varphi$ does not contain quantifiers.
In other words, BCE simulates UDE under the above conditions.
%BCE simulates the removal of those unused definitions $\mathit{def}(P,\varphi)$ for which the formula $\varphi$ does not
%contain a quantifier.

%Then, the question whether BCE simulates UDE, symbollically expressed as \[\mathit{BCE}(\cnf{}(F)) \subseteq \cnf{}(\mathit{UDE}(F)),\]
%does not only depend on the input formula $F$, but also on the used clausification procedure $\cnf{}$.
%%Without specifying the details here, 
%We conjecture that for ``reasonably behaved'' clausification procedures, such as the well-known Tseitin or Plaisted-Greenbaum translation,
%BCE simulates the removal of those unused definitions $\mathit{def}(P,\varphi)$ for which the formula $\varphi$
%does not contain a quantifier.
%
The main idea behind the % conjectured 
simulation would be to show that
each clause stemming from the clausification of an unused definition $\mathit{def}(P,\varphi)$
is blocked on the literal corresponding to predicate $P$.
Although further intuitions
% behind the positive part of our conjecture 
are omitted here due to lack of space,\footnote{
See Appendix~\ref{sec_ude} for an additional discussion.}
the reason why the presence of quantifiers in the definition formula $\varphi$ poses a problem
can be highlighted on a simple example:
\begin{example}
The predicate definition $\mathit{def}(P,\exists x. Q(x)) = P \leftrightarrow \exists x. Q(x)$
can be clausified as $\neg P \lor Q(c),\ P \lor \neg Q(x)$, where $c$ is a Skolem constant corresponding to the existential quantifier.
By resolving these two clauses on $P$ we obtain the resolvent $Q(c) \lor \neg Q(x)$ which is not valid.
\end{example}

%-----------------------------------------------------------------------------------------------------------
\subsection{Experimental Evaluation}
\label{sec:experiments}
%-----------------------------------------------------------------------------------------------------------

We present an empirical evaluation of our implementation of blocked-clause elimination,
which is part of the preprocessing pipeline of the automated theorem prover \vampire{} \cite{KovacsVoronkov:CAV:Vampire:2013}.
In our experiments, we used the \num{15942} first-order benchmark formulas of the latest TPTP library \cite{TPTP} (version 6.4.0). 
Of these benchmarks, \num{7898} were already in CNF, 
while the remaining \num{8044} general formulas needed to be clausified by \vampire{} before being 
subjected to BCE.
This clausification step was optionally preceded by \vampire{}'s implementation of \ppe{} and \ude{} (see Section~\ref{sec:related}).
\SI{73}{\percent} of the benchmark formulas contain the equality predicate.
In these formulas, we eliminated equality-blocked clauses while in the others we eliminated blocked clauses.
All experiments were run on the StarExec compute cluster~\cite{starexec}.

%With our experiments, we aim at answering two questions:
%\begin{inparaenum}[(1)]
%	\item Do blocked clauses occur in first-order formulas?
%	\item If yes, how does their removal affect the behavior of 
%first-order \mbox{theorem provers}?
%\end{inparaenum}

\paragraph{Occurrence of Blocked Clauses.} 

Within a time limit of \SI{300}{\second} for parsing,
clausification (if needed), and subsequent blocked-clause detection and elimination 
our implementation was able to process all but one problem. % SYO599-1.p
Average/median time for detecting and eliminating blocked clauses was \SI{0.238}{\second}/\SI{0.001}{\second}.
\iffalse
\footnote{
We observed two main classes of problems on which BCE took non-negligible time. 
One is the Cyc ontology problems from the \texttt{CSR} domain, which typically contain in order of millions of clauses
and take around a minute to just parse. Subsequent BCE took under \SI{15}{\second} on each.
The second is a small set of problems of the \texttt{KRS} and \texttt{SYO} domains 
obtained by a tool generating EPR Formulas from QBF \cite{DBLP:conf/cade/SeidlLB12}.
These problems contain predicates with large arities (more than 100 arguments) filled up mostly with variables
and pose an efficiency bottleneck to our current implementation,
because most of unification tests succeed and produce large unifiers as a result.
The single problem for which BCE did not terminate within the time limit belongs to this category.}
\fi

% UDE does not interfere on the "-" benchmarks, i.e.\ those already in CNF.

\begin{figure}
\centering
\begin{minipage}{.49\textwidth} %
\includegraphics[width=\textwidth]{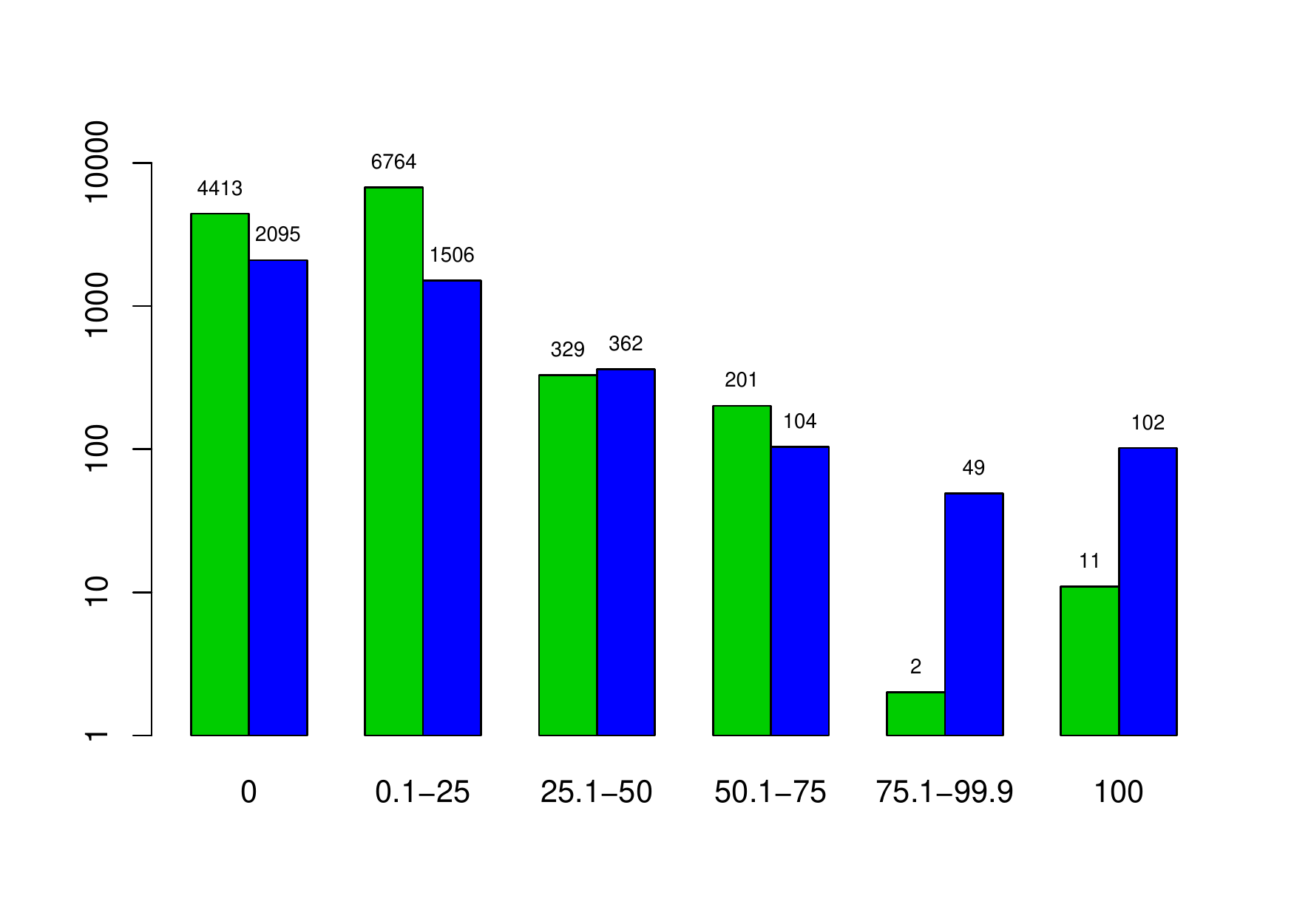}
\end{minipage} %
\begin{minipage}{.429\textwidth} %
\includegraphics[width=\textwidth]{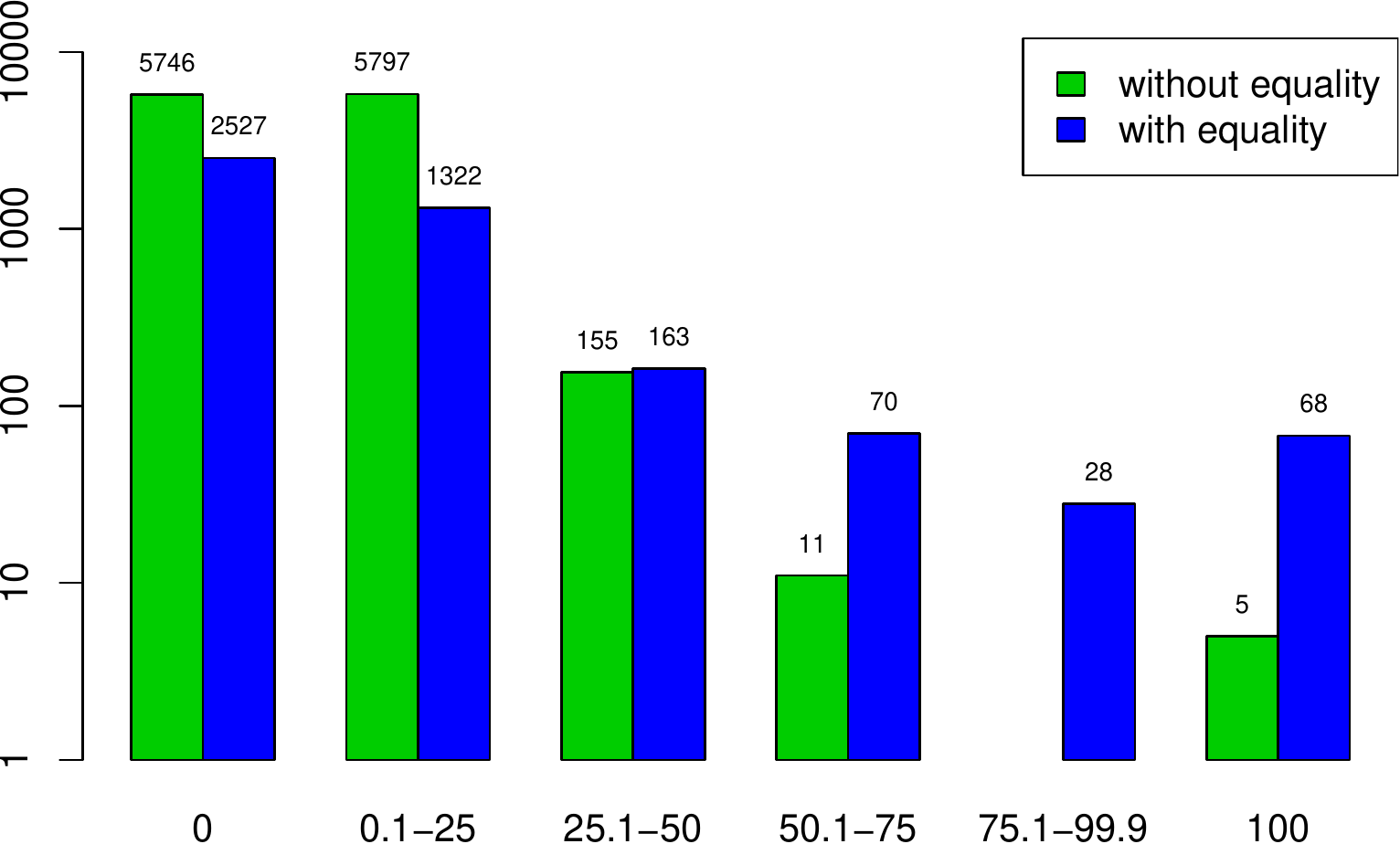}
\end{minipage}
\caption{Histogram of problems by the percentage of blocked clauses eliminated. 
BCE applied to formulas after simple clausification (left) and after clausification preceded by \ppe{} and \ude{} (right). }
\label{fig:dist}
\end{figure}
% The numbers do not sum up to 15941, since they do not contain problems which were already empty before BCE started
% (there were 3 such problems for histogram_bce_noupdr and 46 additional ones for histogram_bce_updr)

%The \num{15941} benchmarks which were fully processed correspond 
In total, the benchmarks correspond to \num{299379591} clauses.
BCE removes \SI{11.72}{\percent} of these clauses, while independently processing the 
problems with \ppe{} and \ude{} before clausification leads to \SI{7.66}{\percent} fewer clauses. 
Combining both methods yields a total reduction of \SI{11.73}{\percent}.
Hence, the number of clauses which can be effectively removed by \ude{} but not by BCE 
or which can only be removed by BCE after some other clauses have been effectively removed by \ude{}
is in the order of \SI{0.01}{\percent}.

Out of the \num{15941} benchmarks, \SI{59}{\percent} contain a blocked clause after simple clausification
and \SI{48}{\percent} of these benchmarks contain a blocked clause if first processed by \ppe{} and \ude{}.
Figure~\ref{fig:dist} shows the detailed distribution of eliminated blocked clauses.
With \ppe{} and \ude{} disabled, more than \SI{25}{\percent} of the clauses could be eliminated 
in over \num{1000} problems. Moreover, \num{113} satisfiable formulas were directly solved by BCE,
which means that BCE rendered the input empty. After applying \ppe{} and \ude{}, which directly solve
46 problems, subsequent BCE can directly solve 73 other problems. There are two problems which
can only be directly solved by the combination of \ppe{}, \ude{} and BCE.

% The following is a decomposition of the whole set of 15941 benchmarks based on the characteristic vector
% (none, PLEUDE, BCE, PLEUDEBCE)
% (True, True, True, True): 3,

% (False, True, True, True): 42,
% (False, True, False, True): 4,

% (False, False, True, True): 71,
% (False, False, False, True): 2
 
% (False, False, False, False): 15819,

% ---

% Note: Some problems are inherently inaccessible to BCE. For instance, TPTP contains 1092 unit equality problems.

\paragraph{Impact on Proving Performance.} To measure the effect of BCE on 
recent theorem provers, we considered the three best 
different\footnote{Actually, Vampire~4.1 was ranked second, but we did not 
include it, as it is just an updated version of 
Vampire~4.0.}
systems of the main FOF division of the 
2016 CASC competition \cite{casc2016}:
%namely the provers 
Vampire~4.0, E~2.0, and CVC4~1.5.1.
%\footnote{
%We left out the actual second best system, Vampire 4.1, because it is just an updated version of Vampire 4.0,
%showing similar behavior, and was only second best in the CASC competition as compared to Vampire 4.0 which was first place.}
% TODO: citations for the provers?
%
Instead of running the provers in competition configurations,
which are in all three cases based on a portfolio of strategies and thus lead to results 
that tend to be hard to interpret (c.f.\ \cite{DBLP:conf/cade/Reger0V14}), 
we asked the respective developers to provide a single representative strategy good for proving theorems by their prover
and then used these strategies in the experiment.\footnote{\label{footnote_strategies}
The strategies are listed in Appendix~\ref{sec_strategies}.}

We combined \vampire{} as a clausifier with the three individual provers using the unix pipe construct.
The clausification included \ppe{} and \ude{} (enabled by default in \vampire{}) 
and either did or did not include BCE. 
We set a time limit of \SI{300}{\second} for the whole combination,
so the possible time overhead incurred by BCE implied shorter time left for actual proving.
We ran the systems on the \num{7619} problems established above on which BCE eliminates at least one clause.
% blocked_set size 7619

% UNSAT
% len(V[2]-V[3]) 28
% len(V[3]-V[2]) 40
% V 3172 3184

% len(E[2]-E[3]) 20
% len(E[3]-E[2]) 27
% E 3097 3104

% len(C[2]-C[3]) 18
% len(C[3]-C[3]) 37
% C 2930 2949

%SAT
%len(V[2]-V[3]) 0
%len(V[3]-V[2]) 5
%V 458 463

%len(E[2]-E[3]) 1
%len(E[3]-E[2]) 9
%E 363 371

%len(C[2]-C[3]) 0
%len(C[3]-C[3]) 68
%C 9 77

%TOTAL
%len(V[2]-V[3]) 28
%len(V[3]-V[2]) 45
%V 3630 3647

%len(E[2]-E[3]) 21
%len(E[3]-E[2]) 36
%E 3460 3475

%len(C[2]-C[3]) 18
%len(C[3]-C[3]) 105
%C 2939 3026

\setlength{\tabcolsep}{3pt}
\begin{table}
\caption{Effect of blocked-clause elimination (BCE) on theorem proving strategies. 
Bold: numbers of solved problems without BCE, 
positive (negative): problems gained (lost) by using BCE.}
%\caption{The effect of BCE on proving theorems. 
%Showing for each prover the number of problems solved without BCE and in the parentheses 
%the number of problems (-) lost / (+) gained by enabling the technique, respectively.}
\label{table:proving_results}
\begin{center}
\begin{tabular}{l||rrr|rrr|rrr}
        & \multicolumn{3}{c|}{unsatisfiable} & \multicolumn{3}{c|}{satisfiable} & \multicolumn{3}{c}{total} \\
\hline
Vampire   & \bf 3172 & $-28$ & $+40$ &
                \bf 458 & $-0$ & $+5$ & 
                \bf 3630 & $-28$ & $+45$  \\
E & \bf 3097  & $-20$ & $+27$ &
       \bf 363 & $-1$ & $+9$ &
       \bf 3460 & $-21$ & $+36$ \\
CVC4 & \bf 2930 & $-18$ & $+37$ &
       \bf 9  & $-0$ & $+68$ &
      \bf 2939 & $-18$ & $+105$ \\
\hline
\end{tabular}
\end{center}
\end{table}

Table~\ref{table:proving_results} shows the numbers of solved 
problems without BCE and the difference when BCE is enabled. 
We can see that on satisfiable problems, 
BCE allows every prover to find more solutions;
the most notable gain is observed with CVC4.
BCE also enables each prover to solve new unsatisfiable problems,
but there are problems that cannot be solved anymore (with the preselected strategy)
 when BCE is activated.
Although the overall trend is that using BCE pays off,
the existence of the lost problems is slightly puzzling.
For a majority of them, the time taken to perform BCE is negligible 
and thus cannot explain the phenomenon. Moreover,
proofs that would make use of a blocked clause,
although they do sometimes occur, are quite rare.\footnote{
For \num{51} of the \num{3172} problems shown unsatisfiable by Vampire, the corresponding proof contained a blocked clause. However,
none of these problems were among the 28 which Vampire did not solve after applying BCE. %(see Table~\ref{table:proving_results})
}
Our current explanation thus
appeals to the inherently ``fragile'' nature
of the search spaces traversed by a theorem prover, in which the presence
of a clause can steer the search towards a proof even if the clause does
not itself directly take part in the proof in the end.

% I checked several lost problems. I could not find a case in which a blocked clause would take part in the found proof.

\paragraph{Strategies for Showing Satisfiability.}
Since the previous experiment indicates that BCE can be especially helpful on satisfiable problems,
we decided to test how much it could improve strategies explicitly designed for establishing satisfiability,
such as finite-model finding. % CITE ?
This should be contrasted with the previous strategies, which focused on showing theoremhood.
Here we selected three systems successful in the FNT (First-order form Non-Theorems) division of the 2016 CASC competition,
namely Vampire~4.1, iProver~2.5, and CVC4~1.5.1 and again picked representative strategies for each,
this time focusing on satisfiability detection.\footnote{
See Appendix~\ref{sec_strategies_sat} for the list of strategies selected for each system.} 
The overall setup remained the same, with a time limit of \SI{300}{\second}.
% TODO: defend why Nitpick was skipped?

\setlength{\tabcolsep}{3pt}
\begin{table}
\caption{Effect of blocked-clause elimination (BCE) on satisfiability checking strategies. 
Bold: numbers of solved problems without BCE, 
positive (negative): problems gained (lost) by using BCE.}
\label{table:satchecking_results}
\begin{center}
\begin{tabular}{l||rrr|rrr|rrr}
        & \multicolumn{3}{c|}{satisfiable} & \multicolumn{3}{c|}{unsatisfiable} & \multicolumn{3}{c}{total} \\
\hline
Vampire   & \bf 531 & $-0$ & $+24$ &
                \bf 719 & $-4$ & $+5$ & 
                \bf 1250 & $-4$ & $+29$  \\
iProver & \bf 558  & $-0$ & $+1$ &
       \bf 755 & $-6$ & $+4$ &
       \bf 1313 & $-6$ & $+5$ \\
CVC4 & \bf 489 & $-1$ & $+28$ &
       \bf 1724  & $-24$ & $+20$ &
      \bf 2213 & $-25$ & $+48$ \\
\hline
\end{tabular}
\end{center}
\end{table}

Table~\ref{table:satchecking_results} provides results of this experiment. We can see that Vampire and CVC4
detected significantly more satisfiable problems when BCE was used. On the other hand, iProver only solved one
extra satisfiable problem with the help of BCE. The results on unsatisfiable problems, which are not 
specifically targeted by the selected strategies, were mixed, not showing a clear advantage of BCE.

\paragraph{Mock Portfolio Construction.} 

Understanding the value of a new technique within a theorem prover is very hard.
The reason is that---in its most powerful configuration---a theorem prover usually employs a portfolio of strategies 
and each of these strategies may respond differently to the introduction of the new technique.
In fact, a portfolio constructed without regard to the new technique is most likely suboptimal
because the new technique may---due to interactions which are typically hard to predict---give 
rise to new successful strategies that could not be considered previously
(c.f. \cite{DBLP:conf/cade/Reger0V14}, Section 4.4).
In this final experiment, we tried to establish the value of BCE for the construction of a
new strategy portfolio in \vampire{} by emulating the typical first phase of the portfolio construction process,
namely random sampling of the space of all strategies.
Encouraged by the previous experiment, we focused on the construction of a portfolio 
specialized on detecting satisfiable problems.

% "Hardness" is defined as
%
%   number of unsuccessful strategies
%--------------------------------------------
%        total number of strategies
%
% among all strategies stored in Spider and run on the problem.

We took a subset of \num{302} % vampire hardness < 1 and hardness >= 0.5
satisfiable problems from the TPTP library that were previously established hard for \vampire{},
% but still solvable,
and that all contain at least one predicate which is different from equality.
We randomly generated strategies by flipping values of various options that define how the prover attempts to establish satisfiability. Each strategy was cloned into two, one running with BCE as part of the preprocessing and the other without. 
% We made sure all strategies use \ppe{} and \ude{}.
Every such pair of strategies was then run on a randomly selected hard problem with a time limit of \SI{120}{\second}. 
In total, we ran \num{50000} pairs.

% The mixed experiment was on a subset of 3,175 with hardness < 1 and hardness > 0.7
% there we had 30s runs
% also \num{50000}

% BCE succeeded 1370+53 times on uns 
% noBCE succeeded 1370+46 times on uns 

% BCE succeeded 25+34 times on sat 
% noBCE succeeded 25+0 times on sat

% for more details see andrei_data/

Strategies using BCE succeeded \num{8414} times while
strategies not using BCE succeeded \num{6766} times.
There were \num{1796} cases where only the BCE variation succeeded on a problem compared to \num{148} cases where only
the strategy without BCE succeeded. This demonstrates that BCE is a valuable addition to the set of \vampire{} options and
will likely be employed by a considerable fraction, if not all,
of the strategies of the satisfiability checking CASC mode portfolio of the next version of the prover.

%The results are presented in Table~\ref{table:proving_results}. 
%Before analysing them one should keep in mind
% Overall, the results are very encouraging. 
% Note that the selected provers are already highly tuned
% in their standard configurations and that the TPTP
% library contains a large number of problems that are either very hard or very easy to solve.
% In both cases, the room for improvement is limited (for an extensive discussion on this issue, see~\cite{DBLP:conf/cade/Reger0V14}). Our experiments clearly indicate that a
% simple technique from SAT solving can be successfully applied in a first-order setting. 

%In other words, despite the high overall number of selected problems 
%the possible room for improvement is rather small.
% 2) Theorem proving is a very ``fragile'' process subject to the butterfly effect.
% This means that even marginal variations in the input or its processing
% may in effect steer the search into completely different direction
% leading to proofs lost or gained randomly -- BUT THEN: how much of the variation is due to BCE and how much is random?

%-----------------------------------------------------------------------------------------------------------
\section{Conclusion}
\label{sec:conclusion}
%-----------------------------------------------------------------------------------------------------------

%This paper shows how to lift blocked clauses to first-order logic.
%Formerly, blocked clauses found their applications in 
%blocked-clause elimination and decomposition---powerful techniques 
%for reasoning  
%in propositional logic, QBF, and DQBF. 
%We implemented our new notion of blocked-clause elimination for first-oder logic in a 
%preprocessing tool which we evaluated in combination with several state-of-the-art theorem provers.
%Our experimental results show
%that a simple SAT technique can also be successfully employed 
%in the first-order setting.

%We lifted blocked clauses to first-order logic 
%and proved their redundancy. 
%We showed that in the presence of 
%equality, a refined notion of blockedness is required to 
%guarantee redundancy.  

We lifted blocked clauses to first-order logic 
and showed that the presence of 
equality requires a refined notion of blocking for 
guaranteeing redundancy. We proved that checking blockedness 
is possible in polynomial time and,  
based on our theoretical results, implemented blocked-clause elimination 
for first-order logic to showcase a practical application of blocked 
clauses. In our evaluation, we observed that the elimination of blocked 
clauses is beneficial for modern provers in many cases, especially when dealing with satisfiable input formulas.

So far, we only investigated the impact of blocked-clause elimination 
as a stand-alone technique. 
From SAT and QBF research, however, it is known that blocked-clause elimination is
even more powerful in combination with other preprocessing 
techniques~\cite{DBLP:journals/jair/HeuleJLSB15} and so we expect this
to be the case in first-order logic too. 
In particular, the combination of variable elimination and blocked-clause elimination
has shown to be very effective in SAT solving~\cite{BlockedClauseElimination}.
It would therefore be interesting to analyze how the combination of first-order
blocked-clause elimination with the predicate elimination technique of
Khasidashvili and Korovin~\cite{khasidashvili16_fo_predicate_elimination} affects
the performance of theorem provers.
%Hence, it would be interesting 
%to lift further techniques from 
%propositional logic and related 
%formalisms~\cite{DBLP:journals/jair/HeuleJLSB15} to first-order logic. 
Moreover, since blocked-clause elimination leads to even greater performance
improvements when used not only before but also during SAT and QBF solving~\cite{DBLP:conf/lpar/LonsingBBES15},
a question arises how to integrate it more tightly into the theorem-proving process. 
Besides elimination, there are other applications for blocked 
clauses as well, like  
the addition of (small) blocked clauses or blocked-clause 
decomposition. We expect that such techniques, for which 
this paper lays the groundwork,  can be helpful 
in the context of first-order theorem proving.
%Finally, our paper lays the groundwork for 
%realizing blocked-clause decomposition also in first-order logic.

%In this paper, we lifted blocked clauses to first-order logic.
%Formerly, blocked clauses found their application in 
% blocked-clause elimination and decomposition---powerful techniques 
% for reasoning  
%in propositional logic, QBF, and DQBF. So far, blocked clauses 
%have not been introduced in first-order logic. 
%Due to peculiarities of first-order resolution as well as of the 
%equality predicate, the naive generalization to first-order logic  does 
%not guarantee redundancy. We showed how to 
%overcome the problems arising from these peculiarities 
%and proved that our generalizations guarantee redundancy.
%We implemented blocked-clause elimination in a preprocessing tool and evaluated it 
%in combination with several state-of-the-art theorem provers. t
%With this work, we narrowed the gap between ``strong'' first-order  theorem provers 
%and ``fast'' SAT solvers by showing---underpinned by extensive experiments---that 
%a simple SAT technique can also be successfully employed 
%in automated theorem proving.  
% 
%In the future, we plan to continue this work by lifting further techniques from 
%propositional logic and related formalisms as presented in~\cite{DBLP:journals/jair/HeuleJLSB15}
%to first-order logic. Further, we considered 
%only preprocessing so far, but we expect further improvements if reasoning on blocked 
%clauses is directly integrated in the solving process~\cite{DBLP:conf/lpar/LonsingBBES15}. 

\section*{Acknowledgements}

We thank Andrei Voronkov for performing the mock portfolio construction experiment.

\newpage
\bibliographystyle{plain}
\bibliography{references}

\newpage

\appendix

\noindent
{\bf \LARGE Appendix}

\section{Polynomial Time Equality-Blocking Check} \label{sec_equational_cplx}

The purpose of this appendix is to argue that the ideas behind Algorithm~\ref{alg_lres_valid} 
presented in Section~\ref{sec:complexity} carry over to the case of equality-blocking, where 
we deal with flat $L$-resolvents and check validity in the presence of equality. Thus we will
also obtain a polynomial time procedure for the equational case.

Because the top-level structure of the procedure along with the main arguments remain unchanged,
we do not repeat them here and instead focus on highlighting the driving analogies 
between the non-equational and equational case. For this purpose let $C = L \lor C'$
be a candidate clause and $D = N_1 \lor \ldots \lor N_n \lor D'$ a partner clause.
We assume, without the loss of generality, that $L = P(\vec{s})$ for a vector of terms $\vec{s}$
and, correspondingly, each $N_i = \neg P(\vec{t_i})$ for a vector of terms $\vec{t_i}$.
A flat $L$-resolvent corresponding to a non-empty set of indexes $I \subseteq \{1,\ldots,n\}$ (i.e., 
a resolvent where the
literals $N_i$ with $i \in I$ are unified with $\bar L$)
can be written as 
\[R_I = C' \lor \bigvee_{i \in I} \vec{s} \foneq \vec{t_i} \lor \bigvee_{i \notin I} N_i \lor D'\]
and it is valid if and only if the ground conjunction of units $\neg \forall R_I$ is unsatisfiable 
in the theory of uninterpreted functions.
% (Let us here ignore the orthogonal issue of variables which should be understood as skolemized in $\neg \forall R_I$ or equivalently as ``rigid'' in $R_I$.)

If we compare how a transition from $I$ to a larger set of indexes $J \supset I$
is reflected on the corresponding (flat) $L$-resolvents, we observe the following. 
In the non-equational case, $R_J$ has fewer literals than $R_I$,
because those corresponding to indexes $J \setminus I$ are missing,
but is obtained using a unifier which is an instance of the one used for obtaining $R_I$.
In the equational case, $R_J$ has analogously fewer literals $N_i$,
but has more literals of the form $\vec{s} \foneq \vec{t_i}$.
From the perspective of the ``complemented'' presentation, 
$\neg \forall R_J$ has fewer atomic literals $p(\vec{t_i})$ than $\neg \forall R_I$,
but has a larger set of equations $\vec{s} \foeq \vec{t_i}$.

We are now ready to describe the analog of Algorithm~\ref{alg_lres_valid} for the equational case.
First, the condition of the while loop (line~\ref{alg_lres_valid:while_loop}) becomes ``constant true'', 
because in the equational case unification can never fail. 
Next, the condition ``$K = \emptyset$'' (line~\ref{alg_lres_valid:checkvalid}),
which corresponds to ``$R_I$ is not valid'',
can be restated as $\neg \forall R_I$ is satisfiable
and decided by a congruence closure algorithm.
Finally, and this is the sole non-trivial part of the analogy,
we need to realize that if $\neg \forall R_I$ is unsatisfiable 
it is either because already
\[\neg \forall \left( C' \lor \bigvee_{i \in I} \vec{s} \foneq \vec{t_i} \lor D' \right)\]
is unsatisfiable and therefore $\neg \forall R_J$ will be unsatisfiable for any $J \supset I$
(this corresponds to the breaking the loop on line~\ref{alg_lres_valid:while_loop_end})
or there is a single index $i \notin I$ such that 
\[\neg \forall \left( C' \lor \bigvee_{i \in I} \vec{s} \foneq \vec{t_i} \lor N_i \lor D' \right)\]
is unsatisfiable and therefore $\neg \forall R_J$ will be unsatisfiable for any $J \supset I$
for which $i \notin J$ (in this latter case, the loop continues as on line~\ref{alg_lres_valid:branch_yes},
with literal $N_i$ added to the set $N$).
This last observation, more specifically
the fact that two distinct literals $N_i = P(\vec{t_i})$, $i \in I$ and $N_j = P(\vec{t_j})$, $j \in I$ 
cannot be both at the same time necessary 
for unsatisfiability of $\neg \forall R_I$ is left as an exercise for the reader.

\section{A Few More Ideas on the Simulation of \ude{} by BCE} \label{sec_ude}

We start by providing a formally more precise definition of UDE taken from \cite{DBLP:conf/fmcad/HoderKKV12}.
Given a predicate symbol $P$, a formula $\varphi$ such that $P$ does not occur in $\varphi$ 
and polarity $\mathit{pol}\in\{0,1,-1\}$,
a predicate definition  is a formula
\[\mathit{def}(\mathit{pol},P,\varphi) = 
\begin{cases}
\forall \vec{x}.\ P(\vec{x}) \leftrightarrow \varphi(\vec{x}), \quad \text{if } \mathit{pol} = 0, \\
\forall \vec{x}.\ P(\vec{x}) \leftarrow \varphi(\vec{x})    , \quad \text{if } \mathit{pol} = 1, \\
\forall \vec{x}.\ P(\vec{x}) \rightarrow \varphi(\vec{x})    , \quad \text{if } \mathit{pol} = -1.
\end{cases}\]
Assuming we have a predicate definition $\mathit{def}(\mathit{pol},P,\varphi)$ 
as a conjunct within a larger formula 
\begin{equation} \label{formula_Psi}
\Psi = \psi \land \mathit{def}(\mathit{pol},P,\varphi),
\end{equation}
UDE allows us 
\begin{inparaenum}[(a)]
\item 
to drop the definition provided $P$ does not occur in $\psi$ or
\item \label{case_2}
to weaken (from an equivalence to an implication) a definition with $\mathit{pol}=0$ to a one with $\mathit{pol}'\in \{1,-1\}$
provided $P$ only occurs with polarity $-\mathit{pol}'$ in $\psi$.
\end{inparaenum}
UDE preserves satisfiability of a formula \cite{DBLP:conf/fmcad/HoderKKV12}.

We assume there is a clausification procedure $\cnf{}$ which takes as an input a first-order formula 
$\varphi$ and transforms it into set of first-order clauses $\cnf[\varphi]$.
The transformation involves operations such as applying de Morgan and distributivity rules,
expanding equivalences, performing skolemization of existential quantifiers and 
naming subformulas to prevent exponential blow-up \cite{DBLP:journals/jsc/PlaistedG86,DBLP:books/el/RV01/NonnengartW01,GCAI2016:New_Techniques_in_Clausal_Form_Generation}. 
Instead of trying to define in the most general terms what 
properties a ``reasonably behaved'' clausification procedure should satisfy
and then showing that our claim holds for any such procedure, 
we present the main ingredients of our argument in a form of an informal proof script. 
A clausification procedure is ``reasonably behaved'' 
whenever this proof script can be used to show our result for it.

Let us now consider formula $\Psi$ as in \eqref{formula_Psi} and
focus on the case \eqref{case_2} of UDE, where a definition with polarity $\mathit{pol}=0$
can be weakened to one with $\mathit{pol}'= 1$, because $P$ only occurs with polarity $-1$ in $\psi$.
The other cases are similar or simpler. The claim that BCE simulates UDE on the CNF-level in this case means 
that there is a sequence of blocked-clause elimination steps turning $\cnf[\psi \land \mathit{def}(0,P,\varphi)]$ to 
$\cnf[\psi \land \mathit{def}(1,P,\varphi)].$ We would like to prove it along the following lines:
\begin{enumerate}
\item \label{item_equiv_expand}
\[\cnf[\psi \land \mathit{def}(0,P,\varphi)] = 
\cnf[\psi] \cup \cnf[\mathit{def}(1,P,\varphi)] \cup \cnf[\mathit{def}(-1,P,\varphi)]\]
using the fact that an equivalence is translated as a conjunction of two implications,\footnote{
And the fact that universal quantifier are simply dropped when transforming a formula to CNF.}

\item \label{item_nonames1}
every $C \in \cnf[\mathit{def}(-1,P,\varphi)]$ is of the form $C = \neg P(\vec{x}) \lor C'$ for $C' \in \cnf[\varphi(\vec{x})]$
using a property of the clausification procedure; moreover, $C'$ does not contain any literal with predicate symbol $P$,
because $P$ does not occur in $\varphi$,

\item \label{item_nonames2}
similarly, every $D \in \cnf[\mathit{def}(1,P,\varphi)]$ is of the form $D = P(\vec{x}) \lor D'$ for $D' \in \cnf[\neg \varphi(\vec{x})]$ and $D'$ does not contain any literal with predicate symbol $P$,

\item \label{item_polarity}
$\cnf[\psi]$ does not contain a clause with a positive occurrence of a literal with predicate symbol $P$ by assumption,

\item \label{item_lemma}
for every $C' \in \cnf[\varphi(\vec{x})]$ and every $D' \in \cnf[\neg \varphi(\vec{x})]$ the clause $C' \lor D'$ is valid.
\end{enumerate}
Finally, the argument would be closed by observing that every $\neg P(\vec{x}) \lor C' \in \cnf[\mathit{def}(-1,P,\varphi)]$ 
is blocked on the literal $\neg P(\vec{x})$, because its only resolution partners are the clauses $P(\vec{x}) \lor D' \in \cnf[\mathit{def}(1,P,\varphi)]$ and each of them leads to a resolvent which is valid by item \ref{item_lemma}.

While items \ref{item_equiv_expand} and \ref{item_polarity} are easy to justify for any reasonable implementation of $\cnf{}$,
formula naming can interfere with the argument behind items \ref{item_nonames1} and \ref{item_nonames2}, and, on top of that,
item \ref{item_lemma} does not work if a skolemisation step needs to be performed 
when clausifying $\varphi(\vec{x})$ or $\neg \varphi(\vec{x})$.
Let us now look more closely at these two caveats.

% TODO: besides the caveats discussed below, it should be clear that everything works also in the presence of equality 
% (the definitions have a nice vector of variables is the argument of the defined predicate).

A clausification procedure may decide to \emph{name} a subformula to prevent, in the worst case, exponential blow-up
stemming from the distributivity of conjunctions over disjunctions (c.f.\ \cite{tseitin_enc}). This is actually achieved
by introducing a new predicate symbol -- the name -- and adding a predicate definition (!) for the name and the subformula.
If a subformula $\chi$ of $\varphi$ is named by $\cnf$, items \ref{item_nonames1} and \ref{item_nonames2} no longer 
hold as stated, because clauses from the definition of $\chi$ will not be of the form $(\neg) P(\vec{x}) \lor C'$,
but rather of the form $(\neg) R(\vec{y}) \lor C'$, where $R$ is the name introduced for $\chi$.
This is ultimately not a problem, since clauses of $\cnf[\mathit{def}(\mathit{pol},R,\chi)]$ will (recursively) become blocked
by the same argument, once $R$ does not occur anywhere else in clausified formula. However, item \ref{item_lemma}
is endangered unless the clausification procedure introduces the same name for $\chi$ a subformula of $\varphi$ and $\neg \varphi$.
The following example illustrates the issue:

\begin{example}
Let us consider the predicate definition $\mathit{def}(0,p,\varphi)$ for $\varphi = (a \land b) \lor c$.
One possible clausification of the definition consists of the clauses: 
\[\cnf_1[\mathit{def}(-1,p,\varphi)]  = \{ \neg p \lor a \lor c, \neg p \lor b \lor c \} 
\text{ and }
\cnf_1[\mathit{def}(1,p,\varphi)] = \{ p \lor \neg a \lor \neg b, p \lor \neg c \}.\]
It is easy to check that, in particular, the clauses $\cnf_1[\mathit{def}(-1,p,\varphi)]$ are blocked on the literal $\neg p$.

Naming the subformula $(a \land b)$ in the definition might lead to the following clausification:
\[ % \cnf_2[\mathit{def}(0,p,(a \land b) \lor c)] = 
\cnf_2[\mathit{def}(0,p,r \lor c) \land \mathit{def}(0,r,a \land b)]= 
\{\neg p \lor r \lor c, p \lor \neg r, p \lor c, \neg r \lor a, \neg r \lor b, r \lor \neg a \lor \neg b\},
\]
in which the clauses from $\mathit{def}(-1,p,r \lor c)$, namely the clause $\neg p \lor r \lor c$,
are blocked on $\neg p$ and after their elimination
the clauses from $\cnf_2[\mathit{def}(-1,r,a \land b)]$, namely $\neg r \lor a$ and $\neg r \lor b$, become blocked on $\neg r$.

However, a clausification procedure which would, for instance, introduce a name for $(a \land b)$ only for the sake of 
$\mathit{def}(-1,p,\varphi)$ and not for $\mathit{def}(1,p,\varphi)$, or vice versa, or which would 
introduce two distinct names and corresponding definitions, one for the positive and one for the negative occurrence of the subformula, would still be correct, but the result could not simplified as claimed above by BCE.
The clausification could then look, for instance, as follows: $\cnf_3[\mathit{def}(0,p,\varphi)] = 
\cnf_3[\mathit{def}(-1,p,r \lor c) \land \mathit{def}(-1,r, a \land b)] \cup \cnf_3[\mathit{def}(1,p,(a\land b) \lor c)]=
\{\neg p \lor r \lor c, \neg r \lor a, \neg r \lor b, p \lor \neg a \lor \neg b, p \lor \neg c\}.$
% notice that eliminating predicate $r$ would solve the problem here
Here, the clause $\neg p \lor r \lor c$ does not resolve to a valid clause with $p \lor \neg a \lor \neg b$.
\end{example}

While it is straightforward to show that item \ref{item_lemma} holds for any reasonable clausification procedure $\cnf$
whenever only propositional rules are applied and formula names, if introduced, are shared between the two polarities of $\varphi$ 
as discussed above,
this item no longer holds when formula $\varphi$ contains a quantifier and a skolemization step becomes necessary
(as already also shown in the main text):
\begin{example} 
Consider a predicate definition $\mathit{def}(0,P,\varphi)$ for the formula: $\varphi(x) = \exists y. Q(x,y)$.
The definition gets clausified as: \[\cnf[\mathit{def}(0,P,\varphi)] = \{ \neg P(x) \lor Q(x,f(x)), P(x) \lor \neg Q(x,y)\},\]
where $f$ is a Skolem function introduced for the sake of the existential quantifier in $\varphi$.
The resolvent $Q(x,f(x)) \lor \neg Q(x,y)$ of the two defining clauses on $(\neg)P(x)$ is not valid
and so the first clause is not blocked on $\neg P(x)$. (Resolving on the second literal leads to a valid clause here, but recall
we do not exclude the possibility of predicate $Q$ occurring also elsewhere in the formula.)
\end{example} 

To sum up, under certain reasonable conditions, which we did not specify formally,
imposed on the clausification procedure,
BCE on the CNF-level simulates UDE from the formula level
provided the definition in question does not contain a quantifier.
It should be clear, on the other hand, that BCE 
can eliminate more than just certain predicate definitions,
simply because the input formula can already be in CNF
to which UDE obviously cannot, in general, apply.

\section{Theorem Proving Strategies Used in the Experiment} \label{sec_strategies}

\subsection{Vampire 4.0}
\texttt{./vampire -t 300 -sa discount -awr 10}
% Just output: -p off -szs on

\subsection{E 2.0}
%\texttt{./eprover -s --simul-paramod --forward-context-sr --destructive-er-aggressive \\ --destructive-er -tKBO6 -winvfreqrank -c1 -Ginvfreq -F1 \\ -WSelectMaxLComplexAvoidPosPred \\ -H'(1.ConjectureGeneralSymbolWeight(SimulateSOS,488,104,105,32,173,0,327,3.6,1.4,1),\\1.FIFOWeight(PreferProcessed),8.Clauseweight(PreferUnitGroundGoals,1,1,0.5),\\3.Refinedweight(PreferGoals,2,4,7,5,6.6),\\2.ConjectureRelativeSymbolWeight(ConstPrio,0.06,67,160,111,25,3.1,2.8,1))'}
\texttt{./eprover -s --simul-paramod --forward-context-sr \textbackslash \\
--destructive-er-aggressive --destructive-er -tKBO6 \textbackslash\\
 -winvfreqrank -c1 -Ginvfreq -F1 \textbackslash\\ -WSelectMaxLComplexAvoidPosPred \textbackslash\\ -H'(1.ConjectureGeneralSymbolWeight\textbackslash\\(SimulateSOS,488,104,105,32,173,0,327,3.6,1.4,1),\textbackslash\\1.FIFOWeight(PreferProcessed),\textbackslash\\8.Clauseweight(PreferUnitGroundGoals,1,1,0.5),\textbackslash\\3.Refinedweight(PreferGoals,2,4,7,5,6.6),\textbackslash\\2.ConjectureRelativeSymbolWeight\textbackslash\\(ConstPrio,0.06,67,160,111,25,3.1,2.8,1))'}

\subsection{CVC4 1.5.1}
\texttt{./cvc4 --full-saturate-quant} 
% Just output: --lang=tptp --no-checking --no-interactive --inst-format=szs --force-no-limit-cpu-while-dump

\section{Strategies for Testing Satisfiability Used in the Experiment}
\label{sec_strategies_sat}

\subsection{Vampire 4.1}
\texttt{./vampire -t 300 -sa fmb} % -sas minisat is default
% -p off -szs on

\subsection{iProver 2.5}
%\texttt{./iproveropt --sat\_mode true --schedule none --sat\_finite\_models true}
\texttt{./iproveropt --sat\_mode true --schedule none \textbackslash\\ --sat\_finite\_models true}
% Just output: --res_out_proof false --inst_out_proof false --sat_out_model none --out_options none

\subsection{CVC4 1.5.1}
%\texttt{./cvc4 --finite-model-find --fmf-inst-engine --sort-inference --uf-ss-fair}
\texttt{./cvc4 --finite-model-find --fmf-inst-engine --sort-inference \textbackslash\\
 --uf-ss-fair}
% Just output: --lang=tptp --no-checking --no-interactive --inst-format=szs --force-no-limit-cpu-while-dump

\end{document}